\documentclass[format=sigconf]{IEEEtran}

\usepackage{algorithm}
\usepackage{algorithmicx}
\usepackage{algpseudocode}
\usepackage{amsmath,amsthm}
\usepackage{amssymb}
\usepackage{array}

\usepackage{enumerate}
\usepackage{fontenc}
\usepackage[bookmarks]{hyperref}
\usepackage{mathtools}
\usepackage{multicol}
\usepackage[subtle]{savetrees}
\usepackage{url}
\usepackage{xspace}
\usepackage{balance}
\usepackage{stmaryrd}
\usepackage[framemethod=tikz]{mdframed}

\usepackage{paralist}

\newtheorem{theorem}{Theorem}
\newtheorem{definition}{Definition}

\newcommand{\myoverset}[2]{\overset{\mbox{{\normalfont\tiny #1}}}{#2}}

\newcommand{\sate}{SAE\xspace}
\newcommand{\sates}{\sate{s}\xspace}

\newcommand{\victim}{alleger\xspace}
\newcommand{\victims}{allegers\xspace}

\newcommand{\accused}{accused\xspace}

\newcommand{\tEscrow}{Escrow\xspace}
\newcommand{\tEscrows}{Escrows\xspace}
\newcommand{\tescrow}{escrow\xspace}
\newcommand{\tescrows}{escrows\xspace}

\newcommand{\shares}[1]{\ensuremath{\llbracket #1 \rrbracket}\xspace}

\newcommand{\ZZ}{\ensuremath{\mathbb{Z}}\xspace}

\renewcommand{\paragraph}[1]{\smallskip\noindent{\bf #1}.}

\newcommand{\Fsate}{\ensuremath{\mathcal{F}_{SAE}}\xspace}
\newcommand{\A}{\ensuremath{\mathcal{A}}\xspace}
\newcommand{\fail}{\ensuremath{\bot}\xspace}

\newcommand{\ensurenotmath}[1]{\ifmmode{\text{#1}}\else{#1}\fi}

\newcommand{\FF}{\ensuremath{\mathcal{F}}\xspace}

\newcommand{\register}{\ensurenotmath{``Register''}}

\newcommand{\ID}{\ensuremath{\mathsf{ID}}\xspace}

\newcommand{\negl}{\ensuremath{\eta}} 

\renewcommand{\sim}{\ensuremath{\mathcal{S}}}
\newcommand{\adv}{\ensuremath{\mathcal{A}}}
\newcommand{\environment}{\ensuremath{\mathcal{E}}}

\title{Finding Safety in Numbers\\ with Secure Allegation Escrows}

\author{
  \IEEEauthorblockN{
    Venkat Arun\IEEEauthorrefmark{1}, Aniket Kate\IEEEauthorrefmark{2}, Deepak Garg\IEEEauthorrefmark{3}, Peter Druschel\IEEEauthorrefmark{3} and Bobby Bhattacharjee\IEEEauthorrefmark{4}\\}
\IEEEauthorblockA{
\begin{tabular}{cc}
  \begin{tabular}{c}
    \IEEEauthorrefmark{1}Massachusetts Institute of Technology \\
    {\tt venkatar@mit.edu}\\
    \IEEEauthorrefmark{3}Max Planck Institute for Software Systems \\
    {\tt \{dg, druschel\}@mpi-sws.org}\\
\end{tabular} & \begin{tabular}{c}
    \IEEEauthorrefmark{2}Purdue University \\
    {\tt aniket@purdue.edu}\\
    \IEEEauthorrefmark{4}University of Maryland\\
    {\tt bobby@cs.umd.edu}\\
  \end{tabular}
\end{tabular}
}
\vspace{-20pt}
}

\begin{document}

\maketitle

\begin{abstract}

For fear of retribution, the victim of a crime may be willing to report it only if other victims of the same perpetrator also step forward. Common examples include 1) identifying oneself as the victim of sexual harassment, especially by a person in a position of authority or 2) accusing an influential politician, an authoritarian government, or ones own employer of corruption. To handle such situations, legal literature has proposed the concept of an \emph{allegation escrow}: a neutral third-party that collects allegations anonymously, matches them against each other, and de-anonymizes allegers only after de-anonymity thresholds (in terms of number of co-allegers), pre-specified by the allegers, are reached.

An allegation escrow can be realized as a single trusted third party; however, this party must be trusted to keep the identity of the alleger and content of the allegation private. To address this problem, this paper introduces Secure Allegation Escrows (SAE, pronounced ``say''). A SAE is a group of parties with independent interests and motives, acting \emph{jointly} as an escrow for collecting allegations from individuals, matching the allegations, and de-anonymizing the allegations when designated thresholds are reached. By design, SAEs provide a very strong property: No less than a majority of parties constituting a SAE can de-anonymize or disclose the content of an allegation without a sufficient number of matching allegations (even in collusion with any number of other allegers). Once a sufficient number of matching allegations exist, the join escrow discloses the allegation with the allegers' identities.  We describe how SAEs can be constructed using a novel authentication protocol and a novel allegation matching and bucketing algorithm, provide formal proofs of the security of our constructions, and evaluate a prototype implementation, demonstrating feasibility in practice.
\end{abstract}

\section{Introduction}

In many cases, the victim or the witness of a crime may be too afraid
to accuse the perpetrator for fear of retribution by the
perpetrator. In other cases, particularly those involving sexual
harassment, the survivor may not report the crime anticipating
negative social consequences or further harassment by the
perpetrator. In such situations, the victim (or the witness) may find
it easier to act against the perpetrator if others also accuse the
perpetrator of similar crimes. Examples of this abound, a notable
example being the recent \#MeToo movement~\cite{metoo}, which led to
many public allegations of sexual abuse in the US film industry and
elsewhere, all triggered by the courage of an initial few.

An \emph{allegation escrow} aids such collective
allegations, by matching allegations against a common perpetrator
confidentially. Technically, an allegation escrow allows a victim or
witness of a crime to file a confidential allegation, which is to be
released to a designated authority once a pre-defined number of
matching allegations against the same party have been filed. The
identities of the accusers and the accused, as well as the content of
the allegation, remain confidential until the release condition holds.

Besides helping fearful victims to report crimes (safe
in the knowledge that their allegation will be revealed only as part
of a larger group), allegation escrows help improve reporting in cases
where the victim is uncertain if the perpetrator's actions constitute
a crime. Escrowed allegations also enjoy higher credibility since, to
all appearances, they are filed independently of each other (as
opposed to public allegations, where the credibility of subsequent
allegations may be questioned). In technical terms, allegation escrows
have been shown to mitigate the {\em first-mover disadvantage} that
perpetrators typically benefit from~\cite{information-escrow-law}.

Project Callisto~\cite{project-callisto} is an allegation escrow system
that has been deployed in 13 universities with over 100k students, to
help report sexual assault on university campuses. A victim can instruct
the system to release the allegation only when another allegation
against the same person exists. Sexual assault survivors who visit the
Callisto website of their university are 5 times more likely to report
the crime than those who do not, and Callisto has reduced the average
time taken by a student to report an assault from 11 to 4
months~\cite{project-callisto-report17}. This makes a very strong case
for the usefulness of allegation escrows.

However, existing allegation escrows such as Project Callisto are
implemented as a \emph{single} trusted third-party, similar to
ombuds-offices in many organizations. Although technically simple and
effective in many cases, the use of a single party may raise concerns
about the escrow's trustworthiness, impartiality and fallibility to
influential perpetrators, thus driving away potential users. In the
case of a university or corporate escrow (e.g. an ombuds-office), students or employees may be
unsure that an allegation against a high-ranking official would be
treated with integrity. A commercial escrow may raise concerns about
its independence from funding sources and long-term security, just as
a government-run escrow may raise concerns about its independence from
high-ups in law enforcement and the judiciary. In all these cases,
users may not trust the escrow enough to file allegations against
people they deem to have the power to coerce or compromise
the escrow. When they do file allegations, strong perpetrators may
actually abuse their power to prematurely discover escrowed
allegations, suppress and alter the allegations, or even seek
retribution against the victims. Finally, even if one victim trusts an
escrow, other victims of the same perpetrator may not, making it
impossible for the escrow to match their allegations.
This suggests the need for allegation escrows based on \emph{several
  independent parties}, none of which in itself is a single point of
coercion or attack by strong adversaries. In this paper, we present
a cryptographic design of such escrows.

\subsection{Contributions}
Our escrows, called \sates
(short for Secure Allegation Escrows), distribute client
secrets---confidential allegations and identities of the alleger and
accused---among several parties by threshold
secret-sharing~\cite{Shamir79}. These parties, called \tescrows, act
together and perform multi-party computations (MPCs) to provide the
same functionality as a single-party allegation escrow, but
compromising less than half of the \tescrows provides no information
about escrowed allegations, accusers or the accused. The \tescrows can
span diverse administrative, political and geographic domains,
mitigating the chances of simultaneous attacks over a majority by the
same adversary. To enable \sate, we make three key technical contributions.

First, \sate needs to provide a strong accountability property: Every filed allegation can be linked to a real-world (strong) identity, which is revealed to the concerned authority once the allegation has found enough matches. This discourages fake allegations and probing attacks all allegation escrows are susceptible to (see~~\S\ref{s:design-requirements}). Simultaneously providing accountability and privacy requires a nontrivial authentication protocol~(see~~\S\ref{proto:reg-and-auth}). When filing an allegation, no minority of \tescrows learn the identity of the filing user. But when the allegation is to be revealed, a majority can determine the identity.

Second, we need to efficiently match allegations to each other, even when each \tescrow only has shares of the allegations, {\it while} providing accountability. For this, we use a novel construction of distributed (verifiable) pseudorandom functions (DVRFs) over shared secrets. This is necessary, since traditional Oblivious PRFs will not provide accountability~(see~~see~~\S\ref{related-work}).

Third, \sate allows each alleger to decide their {\em reveal threshold}; how many matching allegations must be available before their allegation may be revealed. A set of matching allegations, $A$, is revealed if and only if all their thresholds are $\le |A|$. For instance, if three matching allegations with thresholds $\{2, 3, 5\}$ are filed, no allegation should be revealed, since 5's threshold isn't met. But if another allegation with threshold $3$ is filed, then the ones with thresholds $\{2, 3, 3\}$ (but not $5$) should be revealed. We design a novel bucketing algorithm to support reveal thresholds $> 2$  efficiently~(see~~\S\ref{proto:matching:thresholding}).

This flexibility is important since one size doesn't fit all. In many cases of sexual misconduct, a small threshold is desired to maximize the probability of a match; indeed Project Callisto which always uses a threshold of two---where an allegation is revealed if another matching one is filed---has demonstrated real-world utility. However, when the perpetrator is powerful, such as an influential politician, \victims may need many more corroborators to get justice while avoiding adverse consequences for themselves. Similarly, when accusing one's employer (or government) of misconduct or corruption, a person risks getting fired or persecuted. Having just 1 or 2 corroborators may not be much better than being alone. A threshold of 50 (or even 500 or 5000) could be more appropriate. \sate's flexibility allows each \victim to tailor the reveal threshold to the semantics of their allegation.

We formally prove the end-to-end security of our \sate cryptographic design in the universal composability (UC) framework~\cite{CanettiUC}. Specifically, we present an ideal functionality which, by definition, captures the expected security and accountability properties of a \sate, and then show that our cryptographic design \emph{realizes} this functionality. We also implement a prototype of  \sate to understand the latency and throughput of the system. We find that our design is efficient enough for typical use-conditions of allegation escrows.

To summarize, the contributions of our work are:
\begin{inparaenum}
\item The concept of \sate, a distributed allegation escrow, that
  is robust to compromise or coercion of minority subsets of
  constituting parties.
\item A cryptographic realization of \sates using verifiable secret sharing (VSS)
  and efficient multi-party computation (MPC) protocols. In particular, new
  protocols for user authentication, matching and bucketing allegations.
\item A formal security analysis of our cryptographic realization.
\item A prototype implementation and empirical evidence of
  reasonable performance in practice.
\end{inparaenum}

\subsection{Related Work}
\label{related-work}

Ayres and Unkovic~\cite{information-escrow-law} discuss the legal and social utility of allegation escrows in encouraging reporting of sexual misconduct. As discussed, Project Callisto~\cite{project-callisto} is a real deployment that uses a single trusted-party escrow for allegations of sexual misconduct, and has demonstrated the utility of such a system in university settings.

{\bf WhoToo}~\cite{popets-work} is a recent work that proposes a secure allegation escrow for allegations of sexual misconduct. Like \sate, it distributes trust among multiple parties using MPC. \sate differs from WhoToo in a two key ways. First, WhoToo forces all allegations to use a global pre-determined reveal threshold. As discussed above, this inflexibility limits its scope of application. To our knowledge, \sate is the first system to allow each allegation to have its own reveal threshold.

Second, if there are $N$ allegations already in the system, to file the $(N+1)^{th}$ allegation, WhoToo needs to perform $O(N)$ cryptographic operations, including $O(N)$ multi-party computations. \sate is more scalable. Its running time is $O(1)$, independent of the number of pre-existing allegations. We compare the compute complexity in detail in section~\ref{s:proto:prf-cost}.  To obtain such efficiency, \sate sometimes reveals which allegations match which others, before their thresholds are met. Nevertheless, an adversary is unlikely to be able to exploit this~(see~~\S\ref{proto:matching:thresholding})

{\bf Project Callisto} has also developed a prototype cryptographic solution to distribute the trust assumptions~\cite{callisto-crypto,callisto-crypto-2}. It uses a (potentially distributed) Oblivious PRF (OPRF) server. Allegers can query the server to learn the (deterministic) PRF of the accused's identity, while the server just learns the alleger's identity (not the accused's). The alleger uploads the PRF to a database server, which compares them, in clear-text, to match allegations. Callisto's security analysis is informal and has a weaker threat model that admits two attacks.

First, the OPRF server learns the alleger's identity. If a perpetrator compromises this server and learns that one of their victims filed an allegation soon after the crime, they may be able to deduce the probable content of the allegation. By contrast, in \sate, no minority set of \tescrows learn the identity of any \victim until enough matches are found.

Second, it doesn't hold allegers accountable, which allows the adversary to probe how many allegations exist against a given person. They may then guess who filed the allegation from context. To mount the attack, they query the OPRF server to learn the PRF of that person's identity, and compromise the database server to learn how many previously filed allegations match this PRF. \sate prevents this attack by ensuring that if a PRF of an accused's identity is computed, the identity of the alleger is irrevocably tied to the allegation. This enforces accountability and disincentivizes such attacks~(see~~\S\ref{s:design-requirements}).

\paragraph{Trusted Hardware} In recent work, Harnik et al.~\cite{me-too-sgx} use a hardware-backed secure enclave (built on Intel SGX) to isolate a fully autonomous, single-party allegation escrow. The ideas can be combined with \sate to obtain a threat model stronger than either: \sate's \tescrows can be hosted in SGX enclaves to provide a second line of defense even when the administrators of a \emph{majority} of \tescrows are acting maliciously.

\paragraph{Generic MPC and Covert Computation} Generic black-box MPC can also be used to solve the problem. However, like WhoToo, it too incurs at-least $O(N)$ cost per allegation, and doesn't scale. Like covert computation~\cite{covert-multi-party, covert-two-party}, \sate hides even the participation of a user in the protocol, revealing the result only if a pre-defined condition is met. However, black-box covert protocols don't scale well to large numbers of users, and require users to be online for matching to occur.

\section{\sate Design}
\label{design-space}

\subsection{Requirements for Secure Allegation Escrow}
\label{s:design-requirements}
A private allegation escrow system should provide the following
security and privacy properties:
\begin{asparaitem}
  \item {\bf Allegation secrecy}. The escrow should hold each
    allegation secret until enough matches are found. An allegation
    should be released only as part of a group of matching
    allegations.
  \item {\bf Alleger anonymity}. Similarly, the
    escrow should hold each alleger's identity secret until enough
    matches are found.
\item {\bf Scalability}. The escrow should scale to many allegers and allegations. In section~~\S\ref{s:sys-security} we discuss why scalability 1) helps avoid crippling DoS attacks, 2) increases probability of a correct match and, 3) enhances privacy by preventing timing side-channel attacks.

\item {\bf Accountability}. Each allegation is bound to a strong,
  real-world identity.  Once a match is found, the real identities of
  the matched \victims are revealed to the designated authority. Accountability discourages fake/bogus allegations, and acknowledges that the primary source of authenticity of an allegation, escrowed or otherwise, is the human backing it.
\end{asparaitem}

All allegation escrows (not just \sates) are fundamentally vulnerable to probing attacks where the adversary files fake probe allegations against itself in the hope of revealing other genuine allegations before sufficiently many genuine matching allegations have been filed. E.g. the adversary may be a guilty perpetrator seeking vengeance or a troll/journalist seeking a story. While the ultimate defense against such attacks lies in preventing this kind of abuse by non-technical means (e.g., by criminalizing probe allegations), \sates aid such defenses through the property of accountability, which ensures that the real-world identities of all \victims, including fake allegers, are revealed to the designated authority after a match. For this to work, we assume the adversary is afraid of law and/or public perception.

Accountability doesn't deter an adversary who knows their allegation (and hence identity) will never be revealed, perhaps because their reveal threshold is too high or their allegation is unlikely to match any others. We ensure a probe is useful for discovering the presence of only those allegations, that would be revealed at the same time as the probe itself (see~~\S\ref{proto:matching:thresholding}). Hence the probe is just as likely to be revealed as the victim allegation.

Additionally, allegation escrows are most useful in asymmetric
situations, where individual \victims are at a disadvantage compared
to the \accused.  Allegation escrows enable the \victims to build
``strength in numbers'' without fear of premature retaliation.
However, the very information held by allegation escrows motivate
powerful attacks against them, since the \accused can gain by learning
about \victims before a large enough group has formed.  Thus,
allegation escrows should {\em expect\/} to be a targeted.
This leads to the following {\bf meta-property}, that spans
the previous properties.
\begin{asparaitem}
\item {\bf Robustness}. The escrow should resist coercion and
  compromisation attacks. It should continue to provide the properties
  above even if some constituent parts are compromised or
  willingly cooperate with the adversary.
\end{asparaitem}

\begin{figure*}[ht]
  \centering
  \includegraphics[width=0.8\textwidth]{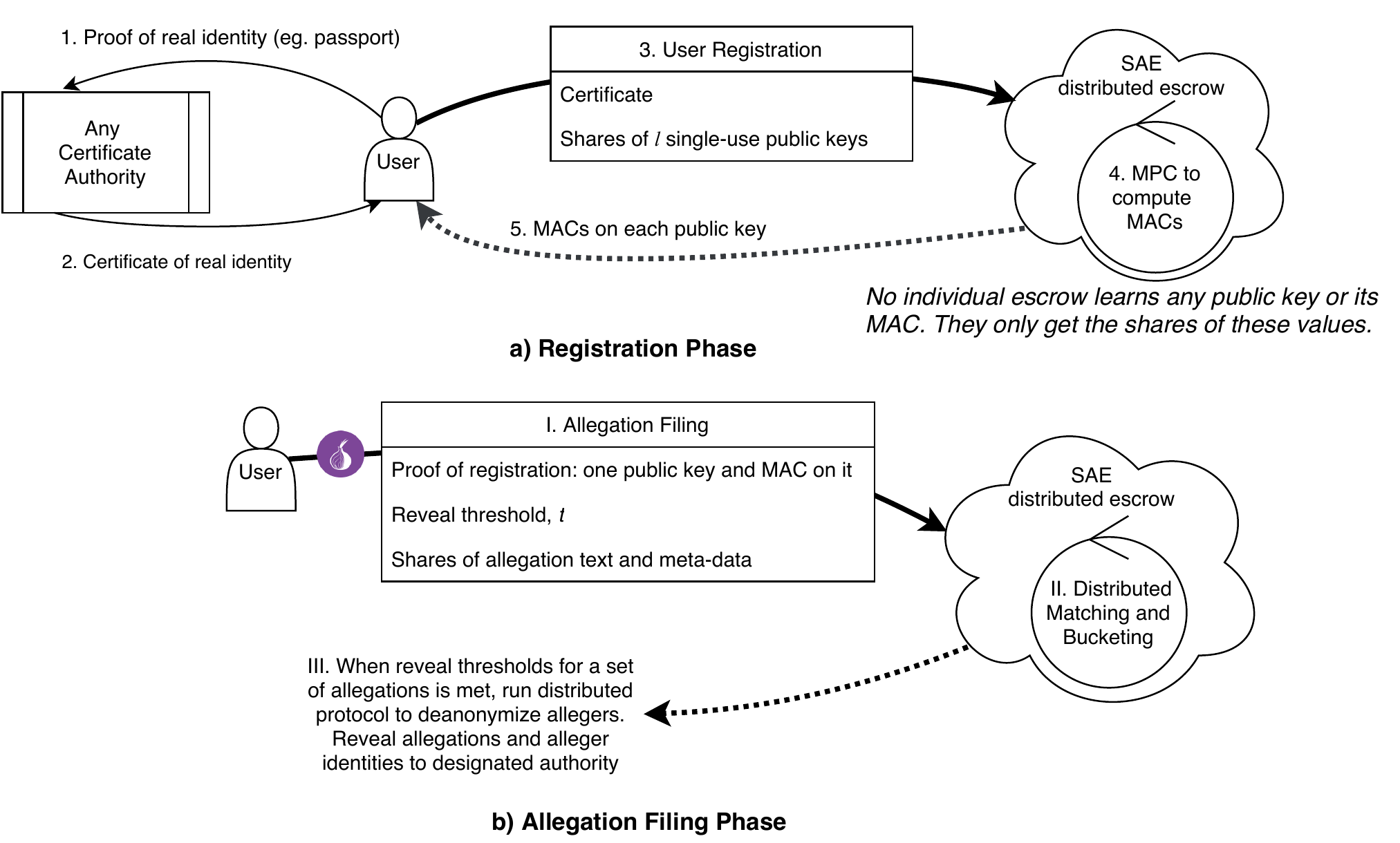}
  \caption{An overview of the \sate protocol. The figure shows a) the user
    registration phase and b) the allegation filing phase. Numbers
    indicate the order in which operations are performed. Thick-continuous and
    thick-dotted lines indicate one-to-many and many-to-one communication respectively. The line in b) with the `Tor' symbol denotes an anonymous channel.}
  \label{fig:overview}
\end{figure*}

\subsection{Threat model and assumptions}
A \sate adversary is interested in prematurely learning the identities
of one or more \victims or discovering unrevealed allegations. For
instance, the adversary may be a guilty perpetrator, interested in
determining whether there is any allegation against them. Or, they may be journalists/trolls looking for a story against a famous person. To this end,
an adversary may {\em actively} compromise some \tescrows into revealing
information they hold and/or not following the \sate protocol
correctly. By design, \sates are robust to such attacks on up to half
the \tescrows simultaneously: allegation secrecy, \victim anonymity, scalability
and accountability hold even if the adversary learns all cryptographic
and allegation-related material possessed by up to half the \tescrows,
and causes them to behave arbitrarily.
Moreover, we expect the compromised
\tescrows to be {\em malicious but cautious}; i.e., we expect the \sate protocol
to catch any malicious behaviour, ensuring that the compromised nodes continue to follow
protocol (say, to avoid detection and removal by the honest majority) and \sate remains live---it continues to offer the expected functionality.

We make the standard assumption that the adversary cannot break
cryptography. Technically, the adversary is a probabilistic polynomial
time (PPT) algorithm with respect to a chosen security parameter
$\lambda$. We assume, as usual, that uncompromised parties (\tescrows
and \victims) keep their long-term secrets safe.

For \victim anonymity, we assume that \victims do not reveal any information beyond that explicitly mentioned in our protocols. For example, they should hide their IP addresses using standard network anonymity solutions like Tor~\cite{tor}. To ensure the time of allegation filing doesn't reveal extra information, honest \tescrows regularly file `garbage' allegations at random times. These are indistinguishable from genuine allegations, and hence serve to hide them. \sate's scalability ensures that this doesn't hurt performance significantly.

\subsection{Protocol Overview}
\label{s:overview}
Figure~\ref{fig:overview} shows high-level protocol flow for the \sate
protocol. We describe the individual protocols for each of these stages in~~\S\ref{s:construction}.

\paragraph{Registration}
\sate uses real-world (strong) identities to ensure accountability.
Prior to registering with \sate, a user proves their real identity to a
certifying authority (CA) and gets its signature on their public
key. The CA may be the user's employer or university registering all
its employees and students into the system, or even an independent
entity verifying physical identities like passports.

To register with a \sate, the user authenticates to all \tescrows
using the CA certificate. The \tescrows and the user then run a
cryptographic protocol during which the user gets
authentication tokens (in particular, MACs) on a fixed number $l$ of fresh public keys. Each of
these $l$ keys can be used once to file an allegation. Importantly, the \tescrows only learn individual \emph{shares}
of these keys, but neither the full keys, nor the MACs on them. This
prevents the \tescrows from learning the identity of a user when the
user files an allegation later, but allows a majority of \tescrows to
reconstruct the identity (by pooling their shares of the public key)
when an allegation has to be revealed.

For their own benefit, users should register ahead of time, even when
they see no need to file an allegation. This prevents timing
correlation channels.  For example, if an \accused is expecting an
allegation due to a recent incident, and colludes with a \tescrow,
then the act of registration by the potential \victim may provide a
strong hint of a pending allegation. Ahead of time registration removes
this channel of inference and could be enforced, for instance, by a
company asking its employees to register with an allegation escrow
service as soon as they join the company

\paragraph{Allegation filing}
When the user wants to file an allegation, they contact the \tescrows,
providing one of the $l$ public keys and the MAC on it, which the
\tescrows can verify. The verification tells each \tescrow that this
user has registered before, but doesn't reveal the
identity of the user, since no \tescrow has seen these in cleartext before. After this, the user provides
the allegation's text along with some meta-data in a specific
cryptographic form, and a \emph{reveal threshold}---the minimum number
of allegations that must match before this one is revealed.

\paragraph{Matching, thresholding and revelation}
The material provided with each allegation is fed into a novel matching and
bucketing algorithm. This algorithm matches allegations to each other. As soon as a set $A$ of matching allegations, each with a reveal
threshold $\le |A|$ (the size of $A$), is found,
they are revealed to a designated authority for further
action. The revelation contains the real identities of the \victims
and the full texts of their allegations. The designated authority can
then take appropriate action.

\section{Threshold Cryptographic Tools}
\label{crypto}

In this section, we present threshold cryptographic protocols that
we use in \sate.  We first describe the necessary threshold
primitives, and then design the required distributed versions
of the signing and private matching protocols.

\subsection{Multi-Party Computation}
\label{crypto:mpc}
An MPC protocol enables a set of parties $\{P_1, P_2, \dots, P_n\}$ to jointly
compute a function on their private inputs in a privacy-preserving
manner~\cite{Yao::82,Chaum::88,Ben-Or::88,Goldreich::87}. More formally, every
party $P_i$ holds a secret input value $x_i$, and $P_1,\ldots,P_n$ agree on some
function $f$ that takes $n$ inputs and provide $y =
f(x_1, \dots , x_n)$ to a recipient while making sure that the following two
conditions are satisfied: $(i)$ {\it Correctness:} the correct value of $y$ is
computed; $(ii)$ {\it Secrecy:} the output $y$ is the only new information that
is released to the recipient.

An $(n, f)$ Shamir secret sharing~\cite{Shamir79} allows a dealer to
distribute shares of a secret among $n$ parties $ \{ P_1, \ldots, P_n\}$
such that any number set of $\le f$ shares reveals no
information about the secret itself, while an arbitrary subset of shares larger than
$f$ allows full reconstruction of the shared secret. Since in some
secret sharing applications the dealer may benefit from behaving maliciously,
parties also require a mechanism to confirm that each $f+1$ subset of shares
combine to form the same value. To solve this problem, Chor {\em et
  al.}\ \cite{cgma85} introduced verifiability in secret sharing, which led to
the concept of {\em verifiable secret sharing} (VSS)~\cite{pedersen-secret,
  Feldman87, pederson-mpc, BackesKP11}.

In our construction we use the MPC protocol by Gennaro \emph{et
  al.}~\cite{pederson-mpc}. It uses VSS, where
Pedersen commitments~\cite{pedersen-secret} on
the Shamir shares are provided to all parties. It works on secrets in
a prime-order ring $\ZZ_q$ and a multiplicative group $\mathbb{G}$ of
order $q$ of size linear in the security parameter $\lambda$.

\subsection{MPC Tools and Notation}
\label{crypto:notation}

\paragraph{VSS and MPC Notation} We denote the $n$ shares of a secret value
$s$ by the set $\shares{s} = \{ \shares{s}_1, \ldots, \shares{s}_n\}$, where
$\shares{s}_j$ represents the VSS share of party $P_j$. In \sates, we use $n = 2f + 1$.

As the employed VSS protocol is additively homomorphic, operations
$\shares{x_1+x_2}_j = \shares{x_1}_j + \shares{x_2}_j$, and
$\shares{cx_1}_j = c\shares{x_1}_j$ for a known constant $c \in \ZZ_q$ can be computed
by each $P_j$ locally using her shares $\shares{x_1}_j,
\shares{x_2}_j$.  The
computation of $\shares{xy}_j$ from given $\shares{x}_j$, $\shares{y}_j$ is
an interactive process and requires cooperation from $2f+1$
parties~\cite{pederson-mpc}. This protocol has identifiable abort~\cite{IshaiOZ14} and can identify which
the non-cooperating parties (if any) are; thus, a malicious-but-cautious party will always cooperate.
We formalize identifiable abort as
${\tt IdentifiableAbort}(i)$, where $i$ indicates that
$P_i$ either offered wrong or no input.\footnote{For MPC with identifiable abort and $2f+1$ nodes, parties can reshare their non-cooperating party's share
with the reduced compromisation  threshold of $f-1$.}

Given threshold addition and multiplication, we can efficiently perform some additional
operations.
In the following, we list the employed VSS and MPC operations. These functions are cooperatively called by each \tescrow with their share of the inputs. When enough \tescrows cooperate, the functions return their values.

\begin{asparaitem}
\item \textproc{VSS}$(x)$: Verifiably secret share $x$ among all the \tescrows such that $f+1$ of them can reconstruct $x$, but no fewer can~\cite{pederson-mpc}. In \sate, $f = \lfloor (n+1)/2 \rfloor$.
\item \textproc{CombineShares}($\shares{x}_j$): Broadcast ${\shares{x}}_j$, gather shares from other $\ge f+1$ parties, and reconstruct the secret $x \in \mathbb{Z}_q$ if at-least $f+1$ honest shares are available.
\item \textproc{RandomCoinToss}(): \:\: Return a share $\shares{r}_j$ of $r \in \ZZ_q$ chosen uniformly at random using
distributed key generation~\cite{DKGJournal,DKG}.
\item \textproc{PublicExponentiate}$(g, \shares{x}_j,$ {\tt recipients}): \:\: Compute $g^x$, where $x$ is secret-shared and $g \in \mathbb{G}$ or $g \in \mathbb{G}_T$ are publicly known generators of bilinear groups~(see appendix~\ref{app:crypto:bilinear}). The result is revealed as clear-text {\it only} to {\tt recipients}, which is a set of parties. In our protocol, {\tt recipients} is either a given client or the set of all the escrows. This operation can be done efficiently with interaction since the result is revealed in the clear, not in secret shared form.
\item \textproc{DVRF}($\shares{SK}_j$, $\shares{x}_j,$ {\tt flag\_proof}, {\tt recipients}): \:\: Return the VRF $F_{SK}(x)$ to {\tt recipients}~(see~~\S\ref{crypto:prf}). If {\tt flag\_proof} is true, also return the proof $\pi_{SK}(x)$, along with the VRF.
\item \textproc{VerifyVRF}$(PK, \pi, x)$: \:\: Verify that $\pi = \pi_{SK}(x)$, where $SK$ is the secret key of the VRF~(see~~\S\ref{crypto:prf}) corresponding to $PK$. Unlike above, this function can be executed locally by each \tescrow without interaction.
\end{asparaitem}

\subsection{Distributed Verifiable Pseudorandom Functions (DVRFs)}
\label{sec:Distributed}
\label{crypto:prf}

\paragraph{VRFs} A verifiable pseudo-random function is a pseudo-random function $F_{SK}(x)$, along with a proof function $\pi_{SK}(x)$. A PPT adversary cannot distinguish $F_{SK}(x)$ from a random function if it doesn't have access to $SK$ or $\pi_{SK}(x)$. However, given $\pi_{SK}(x)$ and a public key $PK$, a PPT can verify that $F_{SK}(x)$ was computed correctly. The formal definition of VRFs is given in Appendix~\ref{app:crypto:vrf}.

In \sate, we need to compute VRFs in a multi-party computation where both the
key and the input values are available in a secret shared form. Any VRF
scheme can be transformed using general purpose MPC to work with shared key and
shared input tags. However, keeping efficiency and practicality in mind, we
choose a VRF construction by Dodis and Yampolskiy from \cite{pairing-vrf}.

In this construction, if a q-Decisional Bilinear Diffie Hellman
Inversion (q-DBDHI) assumption holds in a bilinear group $\mathbb{G}$
with generator $g$~(see Appendix~\ref{app:crypto:bilinear}), then
\begin{equation}\label{eq:PRF}
  F_{SK}(x) = e(g, g)^{1/(x+SK)}
\end{equation}
is a PRF. When coupled with a proof $\pi_{SK}(x) = g^{1/(x+SK)}$, it is a VRF.
Here, $SK$ is a private key chosen randomly from $\mathbb{Z}_q$, and $PK =
g^{SK}$. To verify whether $y = F_{SK}(x)$, we can test whether $e(g^x\cdot PK,
\pi) = e(g, g)$ and whether $y = e(g, \pi)$.

\paragraph{Distributed Input VRF} We need a distributed protocol for computing a VRF. However, we could not use distributed VRF (DVRF) schemes in~\cite{dprf-1,dprf-2,dprf-3} as, in \sate, the VRF computing parties (the \tescrows) know the input only in a secret-shared formed (this will become clear in \S\ref{s:construction}). So, we design a DVRF with secret-shared (or distributed) input messages. Our construction may be of independent interest to other distributed security systems.

A set of $2f + 1$ \tescrows can efficiently compute $F_{SK}(x)$ or $\pi_{SK}(x)$ if each has a share of $x$ and $SK$ as shown in Algorithm~\ref{protocol:algo-pseudorandom-computation}. Here the result is sent only to {\tt recipients} (either {\it all-escrows} or the given {\it client}). If {\tt flag\_proof} is false, \textproc{DVRF} computes $F_{SK}(x) = e(g,g)^{1/(x+SK)}$. Else, it computes $\pi_{SK}(x) = g^{1/(x+SK)}$. $g \in \mathbb{G}$ is a group generator. Given $\pi_{SK}(x)$, the recipient can compute $F_{SK}(x) = e(\pi_{SK}(x), g)$. If an \tescrow refuses to cooperate, the other \tescrows can determine the identity of the corrupted escrow.

\begin{algorithm}
  \caption{Efficient MPC algorithm to compute DVRFs.}
  \label{protocol:algo-pseudorandom-computation}
  \begin{algorithmic}
    \Function{\textproc{DVRF}}{$\shares{SK}$, $\shares{x},$ {\tt flag\_proof}, {\small \tt recipients}}
    \State{$\shares{t_1} \longleftarrow \shares{SK} + \shares{x}$}
    \State{$\shares{blind} \longleftarrow$ \Call{RandomCoinToss}{}()}
    \State{$(\llbracket t_2 \rrbracket, {\tt IdentifiableAbort(i)}) \longleftarrow \shares{t_1} * \shares{blind}$}
    \State{$t_2 \longleftarrow$ \Call{CombineShares}{$\shares{t_2}$}}
    \State{$\shares{exp} \longleftarrow t_2^{-1} * \shares{blind}$}
    \If{{\tt proof\_flag} = {\tt True}}
    \State \Call{PublicExponentiate}{$g, \shares{exp},$ {\small \tt recipients}}
    \Else
    \State \Call{PublicExponentiate}{$e(g,g), \shares{exp},$ {\small \tt recipients}}
    \EndIf
    \EndFunction
\end{algorithmic}
\end{algorithm}

Algorithm~\ref{protocol:algo-pseudorandom-computation} first inverts
$\shares{x}+\shares{SK}$ which takes two multiplications, and then exponentiates
it. The only values available in clear-text (i.e., not information-theoretically
hidden by the secret sharing) are $t_2$ and the final output. $t_2$ is uniformly
distributed and independent of the input, since it is blinded. Hence this
algorithm does not reveal any information about the inputs beyond what is
revealed by the output.

\section{\sate Construction}
\label{s:construction}

In this section, we present the detailed cryptographic protocols we use to
implement a secure allegation escrow. A formal summary of the protocol is given in Figure~\ref{fig:sate-real}.

\subsection{Format of an Allegation}
\label{proto:alleg-format}

An allegation escrow must have some mechanism to determine whether or
not two allegations match. To allow this, along with free-form text
describing their allegation, \victims provide structured meta-data
describing the allegation. \tEscrows deem that two allegations match
if their meta-data are identical. Although simple, this mechanism is
quite effective---it is also used in
Callisto~\cite{project-callisto}, a deployed (non-cryptographic) escrow. Matching should be unambiguous, since false positive matches can cause allegations to be revealed prematurely. Unlike more sophisticated matching criteria, equality tests are simple and robust.

Allegation meta-data is a formatted string containing specific
fields. For instance, it could contain: 1) identity of the \accused
and, 2) the type and intensity of a crime. The identity can be
specified either as a name or as a unique identifier, if available. In
an institutional setting for instance, the user could select from a
drop-down list of other employees/students in that institute. The
`type and intensity' of crime is selected from a drop-down list
containing entries like `sexual harassment', `sexual assault', `petty
theft', `fraud ($<\$10^3$)', `fraud ($\ge \$10^3, < \$10^6$)', `fraud
($> \$10^6$)' and `racial discrimination by a person in power'.


Along with the meta-data and free-form text, the user also submits a
reveal threshold---the lowest number of matching allegations that must
be revealed along with this (or before) this one. Unlike prior
work~\cite{callisto-crypto, me-too-sgx, popets-work}, which only supports a {\it single}
matching threshold throughout the system, we allow the user to pick a
threshold to their own satisfaction with each allegation.

\subsection{Initialization}
All \tescrows register with a standard PKI. They use this to form secure, two-way authenticated TLS links between every pair of \tescrows. These are used for all inter-\tescrow communication. During both registration and filing, the user and \tescrows use a session ID to ensure all all \tescrows are talking to the same user.

The \tescrows use \textproc{RandomCoinToss}() to generate individual shares of private DVRF keys that are later used to 1) register and authenticate users $SK_I$, 2) reveal user identities when required $SK_R$ and 3) match allegations in each bucket $i$, $SK_i$. Since there are infinitely many buckets $i \in 0,1,\ldots$, $SK_i$ is generated lazily when required. How these keys are used will be explained later. The public component of $SK_I$, $PK_I$ is also generated using \textproc{PublicExponentiate} and publicly published. All shares use a fixed recombination threshold of $f + 1 = \lfloor \frac{n+1}{2} \rfloor$, so all the escrows must cooperate to perform operations with these keys, and any minority can be compromised by an adversary without violating any of \sate's properties.

\subsection{User Registration, Allegation Filing and Revelation}
\label{proto:reg-and-auth}

\paragraph{Registration}
The user first obtains a certificate of real identity (e.g., their passport/employee ID) from an appropriate certificate authority (e.g., their employer). This authority is trusted to verify the identity of the user in the real world, denoted as {\tt ID}. During registration, user forms a two-way authenticated TLS link with each \tescrow using this certificate, and non-repudiably signs all communication during registration.

The user generates $l$ random one-time public-private key pairs and secret shares the public parts, $pk_1, \ldots, pk_l$ among the \tescrows. Each of these can be used to file one allegation later.

Using the MPC protocol for \textproc{DVRF}, the \tescrows compute a MAC on each of these public keys $pk_i$, as $(F_{SK_I}(H_2(pk_i)), \pi_{SK_I}(H_2(pk_i))$ using their secret-shared private DVRF key generated at initialization, $SK_I$. $H_2$ is a collision resistant hash function from the set of public keys to $\mathbb{Z}_q$. Each \tescrow learns only its share of the public key and its share of the computed MAC, while the registering user (and not any of the \tescrows) learns the full MAC.

The \tescrows also compute a PRF $F_{SK_R}(H_2(pk_i))$ using a different private key $SK_R$. Individual \tescrows learn the PRF, but nothing else. Each \tescrow stores the association between the user's real-world identity and $F_{SK_R}(H_2(pk_i))$ in a local map. This association is used when revealing allegations later.

At the end of the registration, every \tescrow knows the real user,
but knows only one share of each of the public keys $pk_i$ the user
provided and one share of the MAC computed on it. Consequently, when
presented with one of these public keys and its MAC later, no minority
of \tescrows can link the key back to a specific registered user.

\paragraph{Allegation filing}
A registered \victim files an allegation by connecting to the
\tescrows over an anonymous communication channel
which is modelled by the functionality $\FF_\mathsf{anon}$. (e.g., see \cite{formal-onion}) It anonymously delivers messages to users in the network.
The \tescrows' identities are authenticated with standard PKI.

During the filing, the alleger chooses a random public key $pk$ from the set previously registered and submits 1) $pk$, 2) $\pi_{SK_I}(pk)$, the \tescrows' MAC on it, 3) the
allegation's full text encrypted with a fresh symmetric key $k$, $Enc_k(a)$ 4) shares of $k$, 5) shares of a hash of the allegation's meta-data $m=H_1(\mathrm{meta-data})$, where $H_1: \{0,1\}^* \rightarrow \mathbb{Z}_q$ is a collision resistant hash, 6) an arbitrary reveal threshold $t$, and 7) signatures on all the above with the private key corresponding to $pk$.\footnote{For the formal security, we demand $Enc_k(\cdot)$ to be non-committing encryption. We define it formally in Appendix~\ref{sec:uc-security-proof}, and refer \cite{Chou2015} for simple construction.}

Since no \tescrow has seen the whole public key $pk$ or the entire MAC
on it before, no \tescrow can link it back to any specific
user. However, all \tescrows can locally verify with \textproc{VerifyVRF} that the MAC on the public key
is legitimate and, hence, that the public key comes from a user who
has previously registered. This verification only requires local
computation by each \tescrow and no MPC, which improves efficiency.

Note that no \tescrow has enough information to reconstruct the
allegation, its meta-data or the identity of the \victim. A
majority must cooperate to reconstruct any of these. This ensures the
properties of allegation secrecy and \victim anonymity
(\S\ref{design-space}), even if a minority of the \tescrows cooperate
with the adversary.

\paragraph{Allegation revelation}
Allegations are matched using a dedicated algorithm by the
\tescrows. The algorithm is described in~\S\ref{proto:matching}.  Once
a majority of \tescrows determine that a set $A$ of matching
allegations can be revealed, i.e., they all have thresholds $\le
|A|$, the \tescrows combine their shares to decode the symmetric keys used to
encrypt the texts of the allegations in $A$. These texts are provided
to a designated authority for further action.

Along with the allegation texts, the \tescrows also reveal the
real-world identities of the \victims who filed $A$. To obtain the
identity of an \victim, the \tescrows compute the PRF $F_{SK_R}(pk)$
(using \textproc{DVRF}, the algorithm described in~\S\ref{crypto:prf}), on the public
key $pk$ the \victim used to file the allegation. Recall that the
\tescrows also computed this PRF when the \victim registered and
mapped the PRF to the victim's identity in a local store. Hence, to
discover the user's identity, they merely need to look up the PRF in
the store. This search is done in clear-text locally by each
individual \tescrow and is efficient.\footnote{Note, that we just use a PRF and not the verifiability property of our VRF here.}

Providing the real-world identities of the matched \victims to the
designated authority allows the authority to reach out to the \victims
and also provides the accountability property from
\S\ref{s:design-requirements}.

\paragraph{Registered public keys must not be used twice}
As just described, after an allegation filed with public key $pk$ has
been matched and revealed, the \tescrows map $pk$ to the strong
identity of the individual. Consequently, the key $pk$ should not be
used to file a second allegation unless the \victim wishes to
de-anonymize itself to the \tescrows. To allow users to file multiple
allegations anonymously, a user registers $l$ different keys during a
single registration. This can be repeated periodically, allowing for
$l$ allegation filings for every user within each period. For
instance, every participating individual may register $10$ public keys
every year, thus allowing every user $10$ allegation filings every
year.

\subsection{Matching and Thresholding}
\label{proto:matching}

Now, we discuss how the \tescrows match allegations to each other and reveal sets of
matching allegations when thresholds are met.

\subsubsection{Matching protocol}
\label{proto:matching:matching}
We describe a simple MPC protocol that matches two allegations when
their meta-data hashes are equal. We start by noting that, by design,
our matching protocol does not allow any minority set of \tescrows to
match two allegations on their own. Recall that each \tescrow receives
only a share of the (collision-resistant) hash of the meta-data, $H_1(m)$, of each
allegation. The shares are randomized, so a minority of \tescrows
cannot check the equality of $H_1(m)$ and $H_1(m')$ using the shares
alone. This property is important, else, an adversary who corrupts a
minority of \tescrows can probe existing allegations to discover if an
allegation against a specific individual exists. They can do this
without any honest parties being aware of such probing.

To compare a set of allegations for equality, all the \tescrows
participate in \textproc{DVRF}~(see~\S\ref{crypto:prf}) to compute a pseudo-random function
$F_{SK}(H_1(m))$ for all allegations in the set. The resulting PRF is
revealed in the clear to all \tescrows, but $SK$ and $H_1(m)$
aren't. $SK$ is a shared secret specially generated for each set of
allegations being compared. The sets are determined by the
thresholding protocol described below.

Since the PRF is bijective when the range of $H_1(\cdot) \in \ZZ_q$,
$H_1(m)$ and $H_1(m')$ are equal if and only if $F_{SK}(H_1(m)) =
F_{SK}(H_1(m'))$. Hence, each \tescrow can locally determine which
allegations match which others (note, $H_1$ is collision-resistant). Further, $F_{SK}(\cdot)$ is a PRF
whose secret-key is not used for any other purpose, so no additional
information about $m$ is revealed. Thus all pairs that match in a set
of $n$ allegations can be computed efficiently in linear-time (constant time per allegation).

\subsubsection{Bucketing Protocol for Reveal Thresholds $> 2$}
\label{proto:matching:thresholding}

\begin{algorithm}[h]
  \caption{
  Rules for the secure thresholding algorithm \textproc{Bucketing}, whose interface is described in \S\ref{s:security-definition}. It reveals a set of allegations if and
    only if all of their thresholds are satisfied by that set.}
  \label{proto:algo:thresh}

  Apply the following rules repeatedly (in any order) till no further rules
  apply. Rules 2,3 and 4 only apply to collections that haven't been revealed.
  \begin{enumerate}
  \item When an allegation with threshold $t$ is filed, it forms a singleton
    collection and is added to bucket $t-1$ (since $t-1$ other allegations must
    match the allegation before it is revealed).
  \item If $Min(A)$ is the smallest bucket occupied by a collection $A$ and every
    allegation in $A$ has a threshold $< Min(A) + |A|$, $A$ is copied to bucket
    $Min(A)-1$. Note that $A$ still occupies the buckets it used to occupy.
    Copying merely adds the collection to a new bucket.

  \item When two collections overlap and occupy the same bucket, and their
  allegations are found to match~(\S\ref{proto:matching:matching}), they coalesce into one
  collection.

  \item When a collection reaches bucket $0$, all of its allegations are revealed
    as described in \S\ref{proto:reg-and-auth}.
  \item If a collection $A$ is revealed, we make sure it occupies buckets $1,
    \ldots, |A|$, even as $A$ grows. This enables future matching allegations to
    be revealed.
\end{enumerate}
\end{algorithm}

The above matching protocol is secure when the reveal thresholds equal 2.
Supporting higher thresholds securely {\it and} scalably requires more work.
A collection $A$ of matching allegations should be revealed when every
allegation in $A$ has a reveal threshold no more than the size of $A$
(written $|A|$). One way to find such collections would be to run the
above matching protocol on the set of \emph{all} allegations
irrespective of their thresholds and then locally determine whether an appropriate set $A$ exists.  However, this design is susceptible
to a probing attack where an adversary interested in probing for the
{\it existence} of a specific allegation, files the same allegation with a
very high threshold. By corrupting just one of the \tescrows, the
adversary could then compare this allegation to all other allegations
in the system, without any risk that its own false allegation would
ever be revealed (since the probe allegation has a very high threshold). To
deter such attacks, we control which allegations
\emph{can} be compared to each other. We ensure that {\em if two
allegations can ever be compared by a minority of escrows, then they
will be revealed at the same time, if at all}. That is, two allegations
can be compared by a minority only if they are waiting for the same
number of matching allegations. Now, if the adversary tries to probe
with a fake allegation, the fake allegation (and hence the adversary's
real-world identity) is exactly as likely to be revealed as the
\victim's actual matching allegation.

Thus, to just learn the {\em number} of allegations against a person, the adversary must risk leaving a non-repudiable paper trail. Additionally, if the adversary is a guilty party seeking to determine the number of escrowed allegations against them, they risk precipitating the revelation of an honest allegation, which may have otherwise remained escrowed forever. We assume the adversary won't take such risks.

Note, the above attack only works for threshold $> 2$. If an honest alleger's threshold is 2, \sate doesn't admit any attacks not present in a single trusted-party implementation, even if the adversary is willing to risk filing a probe allegation. Prior work on single trusted-party based allegation escrows~\cite{project-callisto} only supports thresholds of 2, and still demonstrates social utility in allegations of sexual misconduct.

It is possible to use generic MPC to avoid such probing attacks. However, the time taken to process one allegation would then increase with the number of allegations already present in the system. By simply filing many `junk' allegations, an adversary can slow down the system till it is no longer useful. Hence scalability is crucial; our matching/bucketing protocol does only a constant amount of work per allegation.

\begin{figure*}\small
  \begin{mdframed}
  \begin{multicols}{2}
    {\bf Escrow Initialization:} Every  $j^{th}$ escrow executes:
    \begin{enumerate}[1.]
    \setlength{\itemindent}{0em}
        \item {\bf } $\shares{SK_R}_j \leftarrow$ \textproc{RandomCoinToss}()
        \item $\shares{SK_I}_j \leftarrow$ \textproc{RandomCoinToss}()
        \item $PK_I \leftarrow$ \textproc{PublicExponentiate}$(g, \shares{SK_I}_j)$
        \item $\shares{SK_i}_j \leftarrow$ \textproc{RandomCoinToss}() $\forall i \in \{0, 1, 2, \ldots\}$
        \\ {\tt \#$SK_i$ is generated lazily when required}
        \item {\small Set} {\tt identities}$\leftarrow$\{\}; {\tt allegations}$\leftarrow$\{\}; {\tt buckets}$\leftarrow$\{\}
    \end{enumerate}

    {\bf Client Initialization and Registration:}
    The client uses {\tt ID}, a certificate previously obtained from a CA, to authenticate to the escrows.  The \tescrow's identities are managed with PKI. The secure authenticated channel is idealized using $\FF_\mathsf{smt}$. Client signs all messages with {\tt ID}.

    The client executes:
    \begin{enumerate}[1.]
    \setlength{\itemindent}{0em}
        \item$(pk, sk) \leftarrow \mathrm{Gen}(1^\lambda)$ {\it \# Signing key-pair}
        \item Broadcast (``Register'', {\tt ID}) idealized using $\FF_\mathsf{B}$
        \item \textproc{VSS}$(H_2(pk))$ among the escrows

        \hspace{-5mm}            Every  $j^{th}$ escrow executes:
        \item  $\bot \leftarrow$ \textproc{DVRF}$(\shares{SK_I}_j, \shares{H_2(pk)}_j,$ {\it True, client})
        \label{fig:sate-real-dvrf-client}
        \item $R \leftarrow$ \textproc{DVRF}$(\shares{SK_R}_j, \shares{H_2(pk)}_j$, {\it False, all-escrows})
        \item {\tt identities}$[R]$ $\leftarrow \ID$

        \hspace{-5mm}            The client executes:
        \item Receive $\pi_{SK_I}(H_2(pk))$ from escrows' \textproc{DVRF} call in step~\ref{fig:sate-real-dvrf-client}
        \item Store $(pk, sk)$ and $\pi_{SK_I}(H_2(pk))$ for future use.\\
    \end{enumerate}

    {\bf Allegation Filing:}
    With allegation text $a$, $m = H_1$(meta-data), reveal threshold $t$, fresh symmetric encryption key $k$ and, $(pk, sk, \pi_{SK_I}(H_2(pk)))$ produced during registration,
   the client connects to each \tescrow over an anonymous communication channel, idealized using functionality $\FF_\mathsf{Anon}$. It signs all communication with $sk$, which \tescrows verify before accepting the input. The \tescrows process allegations serially. They identify an allegation by the $pk$ used to file it. The client executes:     \begin{enumerate}[1.]
    \setlength{\itemindent}{0em}
        \item Broadcast (``File'', $pk$, $t$), idealized using $\FF_\mathsf{B}$
        \item Broadcast $(\pi_{SK_I}(H_2(pk)), Enc_k(a))$, idealized
        using $\FF_\mathsf{B}$
        \item \textproc{VSS}$(m)$, \textproc{VSS}$(k)$ among the escrows\\

        \hspace{-5mm}       Every  $j^{th}$ escrow executes:

               \hspace{8mm}{\it \# Identity Verification and Filing}
\item If  $\textproc{VerifyVRF}(PK_I, \pi_{SK_I}(H_2(pk)), H_2(pk))$ fails, abort
\item If {\tt allegations}$[pk]$ exists, abort
\item{\tt allegations}$[pk]$ $\leftarrow (t, pk, Enc_k(a), \shares{m}_j, \shares{k}_j)$
\item$M \leftarrow $ \textproc{DVRF}($\shares{SK_{t-1}}_j, \shares{m}_j,$ {\it False, all-escrows})
\item{\tt buckets[$t-1$]} $\leftarrow$ {\tt buckets[$t-1$]}$\cup$ $\{(pk, M)\}$

        \hspace{8mm}{\it \# Matching and Bucketing}
\item$T \leftarrow$ \textproc{Bucketing}({\tt buckets})
\item{\bf While $T \neq \bot$}
\item\quad {\tt \# Move allegation $T.id$ to bucket $T.i$}
\item\quad $\shares{m}_j \leftarrow$ {\tt allegations[$T.id$]}$.\shares{m}_j$
\item\quad $M' \leftarrow$ \textproc{DVRF}($\shares{SK_i}_j, \shares{m}_j,$ {\it False, all-escrows})
\item\quad {\tt buckets[$T.i$]} $\leftarrow$ {\tt buckets[$T.i$]} $\cup$ \{($T.id$, $M'$)\}
\item\quad $T \leftarrow$ \textproc{Bucketing}({\tt buckets})

        \hspace{8mm}{\it \# Reveal Allegations}
\item{\bf For} $(id, M)$ {\bf in} {\tt buckets}$[0]$
\item\quad {\tt buckets}$[0] \leftarrow\,$ {\tt buckets}$[0]$ \textbackslash $(id, M)$
\item\quad $(t, pk, Enc_k(a), \shares{m}_j, \shares{k}_j) \leftarrow$ {\tt allegations}$[id]$
\item\quad $R \leftarrow$ \textproc{DVRF}($\shares{SK_R}_j, H_2(pk),$ {\it false, all-escrows})
\item\quad $\ID \leftarrow$ {\tt identities}$[R]$
\item\quad $k \leftarrow \textproc{CombineShares}(\shares{k}_j)$
\item\quad $a \leftarrow Dec_k(Enc_k(a))$
\item\quad Output the revealed allegation $(t, a, \ID)$
    \end{enumerate}
  \end{multicols}
  \end{mdframed}
  \caption{The \sate protocol. The client has a certificate of identity, $ID$, from a certificate authority. MPC, cryptographic and communication primitives used in this protocol are defined in~~\S\ref{crypto:notation}. For ease of exposition, we only show the registration of one key. In practice, the client can register $l$ keys, and randomly pick one to use during filing.
  We formalize the functionalities $\FF_\mathsf{B}$ and $\FF_\mathsf{Anon}$ in Section~\ref{s:security-definition}. \textproc{Bucketing} is defined in~~\S\ref{proto:matching:thresholding}. }
  \label{fig:sate-real}
\end{figure*}

To keep track of how many matches each allegation needs, each \tescrow
independently maintains buckets numbered $0$, $1$, $2$, $3$,
\ldots. An allegation is in the $i^{th}$ bucket only if it is waiting for $\le i$ more allegations. An allegation may be present
in more than one bucket. Bucket $0$ contains a list of allegations
that have been revealed. Algorithm~\ref{proto:algo:thresh}
controls which allegation occupies which buckets.

Only allegations \emph{within} a bucket may be compared to each other. To ensure this, each
bucket $i$ is associated with an independently chosen secret key $SK_{i}$, which is shared
among the \tescrows ($SK_{i}$ is generated lazily when bucket $i$ is
first used). When an allegation is added to bucket $i$, the \tescrows
compute $F_{SK_i}(H_1(m))$ for that allegation using \textproc{DVRF}~(see~\ref{crypto:prf}). Any \tescrow can use this to locally compare any two allegations in bucket $i$. Since, by design, $SK_i \neq SK_j$ if
$i \neq j$, $H(m)$ and $H(m')$ cannot be compared using
$F_{SK_i}(H(m))$ and $F_{SK_j}(H(m'))$ when $i \not= j$. Allegations
that are known to match each other, either directly because they are
in the same bucket or indirectly by transitivity, are said to belong
to the same `collection'. When allegations from two different
collections are found to match, the collections coalesce into one.
The resulting collection spans the union of buckets spanned by the
parent collections and contains the union of allegations. Every
allegation belongs to exactly one collection at any given time. To
copy all allegations in a collection into a new bucket, the PRF for
only one allegation's meta-data needs to be computed, since all
allegations in a collection have identical meta-data.

This algorithm trivially satisfies the property that, once two allegations are
known to be equal to each other, they belong to the same collection and are
revealed together (if at all). This deters the probing attacks described above
that motivated this elaborate mechanism. We also prove that the thresholding
algorithm is `correct':

\begin{theorem}[Correctness] Algorithm~\ref{proto:algo:thresh}  reveals a
collection if and only if the thresholds of all allegations in it are satisfied.
\label{thm:thresh-correct}
\end{theorem}

\begin{proof}
Let $Max(A)$ and $Min(A)$ be the largest and smallest buckets occupied
by collection $A$.  We begin by proving that the following three
properties hold whenever all five rules of
Algorithm~\ref{proto:algo:thresh} have been applied to saturation
(meaning no further rule applies). (1) every collection spans a
contiguous range of buckets, (2) every collection $A$ spans $|A|$
buckets, i.e. $|A| = Max(A) - Min(A) \myoverset{def}{=} Span(A)$, (3)
every allegation in a collection $A$ has a threshold $\le |A| +
Min(A)$ and hence can be revealed if $Min(A)$ more matches are
available.

The first property can be proved as an invariant that is trivially maintained by
rules 2, 4 and 5 with rule 1 as the base case. Now, two collections coalesce
only if they share a bucket (and hence their allegations may be compared). Since
the union of contiguous, overlapping segments is contiguous, rule 3 also
maintains the invariant.

To prove the second property, note that in any collection $A$, all allegations
have a threshold $\le Max(A)$ by definition. If $Span(A) < |A|$, $Max(A) =
(Max(A) - Min(A)) + Min(A) < |A| + Min(A)$, since $Max(A) - Min(A)
\myoverset{def}{=} Span(A)$. Hence rule 2 can be applied repeatedly until
$Span(A)$ increases to equal $|A|$. Hence $Span(A) \ge |A|$. We now prove that
$Span(A) \le |A|$ is an invariant with rule 1 as the base case.  Rules 4 and 5 trivially maintain the invariant. Rule 2 would not apply if it causes
the invariant to be broken, as there is at-least one allegation with threshold
$Max(A)$ if A is not yet revealed (which is when rule 2 applies). The threshold
condition for this allegation will not be met if $Span(A) > |A|$, as it
implies the threshold $t = Max(A) \myoverset{def}{=} Min(A) + Span(A) > Min(A)
+ |A|$. Applying Rule 3 to create $C$ out of $A$ and $B$,
maintains the invariant. $|C|$ spans a union of the parent's buckets, hence
$Span(C) \le Span(A) + Span(B) \le |A| + |B| = |C|$ because $A$ and $B$ are
disjoint. Hence the invariant is maintained.

The third property is explicitly maintained as an invariant by rule 2
and is trivially satisfied by rules~1 and~5. Rule~3 is applicable in
two ways. First, when a new allegation arrives in between an older
collection, the property is not broken. Second, if two
existing collections, $A$ and $B$, coalesce into $C$ by rule 3, one is
`above' another. Let $Min(B) = Max(A)$, without loss of
generality. Then, $Min(C) = Min(A)$, hence allegations in $A$ satisfy
the property. The drop in $Min$ for allegations in $B$ is $Min(B) -
Min(C) = Max(A) - Min(A) = Span(A) \le |A|$ (we proved above that all rules maintain $Span(A) \le |A|$ as an invariant). This drop in $Min$ is compensated by a
corresponding increase in size of the collection by $|A|$.

We now use these properties to prove correctness. The third property
implies that when a collection $A$ is revealed, the
threshold condition is satisfied for all revealed allegations, since $Min(A)=0$.  To
prove the other direction, let there be $n$ matching allegations such that all
their thresholds are $\le n$. Assume for contradiction that they are not revealed. This
means that they all belong to buckets $1, \ldots, n-1$. By the pigeonhole
principle, there will be one bucket with multiple allegations which
will start coalescing with rules 2 and 3. If the process stops with a collection
of size $k < n$, $n-1-k$ buckets will be left with $n-k$
allegations, because property 2 ensures the size of a collection equals its span. Again, by the pigeonhole principle, the coalescing
process starts. This continues till there is only one collection with
$n$ allegations that spans buckets $0 \ldots n-1$ and all $n$
allegations get revealed. Hence, a set of $n$ matching allegations are
revealed if and only if all their thresholds are $\le n$.
\end{proof}

{\bf \textproc{Bucketing}({\tt buckets}): Interface to the Algorithm} The real protocol~(Figure~~\ref{fig:sate-real}) and ideal protocol~(Figure~~\ref{fig:ideal-func}) interface the bucketing algorithm with the {\textproc{Bucketing}} function. It takes as input the set of buckets, {\tt buckets}. Each bucket $i$, (denoted as {\tt buckets[$i$]}) is a set of tuples $(id, M)$ describing the allegations in that bucket. $id$ is a unique allegation identifier\footnote{In the real protocol, $pk$ can be the allegation identifier, since each key is used only once. For ease of exposition, each user can register for and file only one (not $l$) allegation in our ideal protocol. Hence we use \ID as $id$. The full protocol can be modeled by each user having $l$ distinct \ID s.}, and $M$ is a representation of the meta-data which can be compared for equality to determine which allegation matches which others {\it within a bucket}. The function \textproc{Bucketing} returns a task $T$ if more rules apply in Algorithm~\ref{proto:algo:thresh}. Else it returns $\bot$. Task $T$ instructs the caller to move an allegation $T.ID$ to a bucket $T.i$. For ease of exposition, when a collection is added to a new bucket, \textproc{Bucketing} produces one task per each allegation. In practice, only one is necessary, since all meta-data are identical (see below). Note, \textproc{Bucketing} doesn't need to implement step 1 of Algorithm~\ref{proto:algo:thresh}, since it is implemented by $\Fsate$. It doesn't maintain its own state.

\subsection{Computation Cost}
\label{s:proto:prf-cost}

The computationally expensive steps in the \sate protocol are the VRF/PRF
computations that require interaction between the \tescrows, since they involve
multiplication and exponentiation over shared secrets. However, the
number of such computations scales very well and doesn't increase with the number
of users as well as the number of allegations filed.

Registering a new user requires two computations for every public
key $pk$ that the user provides, one to compute the MAC on the key, and one
to compute $F_{SK_R}(pk)$. Filing an allegation does not, of itself, require any PRF computation.
The cost to move a collection of allegations to a new bucket can be reduced by observing that, we only need to compute PRF for one of the allegations, since
they all have identical meta-data. We can prove that,
in an amortized sense, we need to compute the {\tt PRF} at most twice
for each allegation, independent of its threshold. Revealing
an allegation requires one more PRF computation (to discover the
identity of the \victim). To summarize, on average, every filed
allegation requires 2 PRF computations if it is never revealed, and 3
PRF computations if it is revealed.

\paragraph{Compute Cost Analysis and Comparison with WhoToo} Since there is no implementation for WhoToo, to compare performance, we count the number of cryptographic operations required for \sate and WhoToo, shown in figure~\ref{fig:eval:num-ops}. To improve scalability, WhoToo offloads some work to an offline precomputation step. But non-trivial `online' computation remains. As discussed above, scalability is essential for a practical allegation escrow.

\begin{figure}
    \centering
    \small
    \begin{tabular}{|l|c|c|c|}
    \hline
    & {\bf \shortstack{SAE \\ \textproc{DVRF}}} & { \shortstack{WhoToo\\(online)}} & { \shortstack{WhoToo\\(precompute)}} \\
    \hline
    \multicolumn{4}{|l|}{\bf Interactive MPC Operations} \\
    \hline

    \textproc{Multiply} & 1 & 1 & $N$ \\
    \textproc{PublicExponentiate} & 1 & $3N + q + 1$ & $q$ \\
    \textproc{RandomCoinToss} & 1 & $q$ & $q + 2$ \\
    \textproc{CombineShares} & 1 & 1 & 1 \\

    \hline
    \multicolumn{4}{|l|}{\bf Local MPC Operations} \\
    \hline

    \textproc{Add} & 1 & 0 & 0 \\
    \textproc{Multiply Constant} & 1 & $N$ & 1 \\
    \hline
    \end{tabular}
    \caption{\normalsize Number of various MPC operations required per allegation per \tescrow for \textproc{DVRF} computation in \sate, compared with WhoToo~\cite{popets-work}. Here $N$ is the number of allegations filed so far, and $q$ is WhoToo's global reveal threshold. Here we only show the cost for WhoToo's matching protocol; secure identity verification is separate. \sate requires at-most five \textproc{DVRF} computations per allegation (amortized), including for identity verification. Note that, WhoToo's {\it per-allegation} complexity is $O(N)$, and hence $O(N^2)$ complexity for $N$ allegations.}
    \label{fig:eval:num-ops}
\end{figure}

\subsection{Why is Efficiency Important?}
\label{s:sys-security}
\subsubsection{Resist DoS attacks}
\label{s:dos}
Like any public service, \sate can be flooded with client requests to mount a denial-of-service (DoS) attack. MPC systems can be particularly susceptible to such attacks, since the client can trigger expensive computations at little cost to itself. \sate has two properties that prevent the adversary from exploiting this asymmetric computation cost. 1) {\bf Bounded MPC:} Each registered real identity can trigger only a bounded number of MPC operations. And 2) {\bf Scalability:} The amount of work required to process a user request doesn't increase with the number of allegations or number of registered users.

{\bf Bounded MPC} holds in \sate, since MPC is only triggered after authentication. Since authentication is a local computation, failed authentications do not cause MPC operations. Further, each user (with a real identity) is only allowed to register a fixed number $l$ of public keys in any issue period and each key can only be used to file one allegation. So, a registered user can cause at most $(2 + 3)l = 5l$ PRF computations in any issue period. For $l = 10$ and an issue period of one year, this amounts to at most 50 PRF computations per year per \emph{real} user, which is an extremely low rate for an effective DoS attack.

WhoToo~\cite{popets-work}, a concurrent work, does not have the {\bf scalability} property. Filing the $N^{th}$ allegation takes $O(N)$ MPC operations, so the total cost for filing $N$ allegations is $O(N^2)$. Asymptotically, this is the same as black-box MPC. Further, it allows users to file an arbitrary number of allegations. There is no obvious way to prevent this with their authentication protocol without breaking the ``{\bf bounded MPC}'' property. Hence an adversary could file a million `garbage allegations' (i.e. against random strings), and prevent proper functioning.

\subsubsection{Larger Allegation Pool}
To maximize the probability that a matching allegation will be found, we must allow a large set of sets to use the same escrow. This allows large pools of allegations to be matched with each other. Scalability is essential to enable this.

\subsubsection{Avoid Timing Side-Channels} If somebody commits a crime and learns (through a compromized \tescrow) that an allegation was filed two days later, when filings are otherwise rare, they may reasonably conclude that their victim filed the allegation. We excluded such side-channels in the threat model. To realize this, honest \tescrows (and other external well-wishers) can regularly file decoy allegations. Since \sate maintains anonymity and privacy, the adversary cannot distinguish decoys from real allegations. Those filing decoy allegations must register their separate (real) identities for doing so, and enter a contractual obligation to ensure no decoy ever gets revealed. This can be ensured by filing allegations against random strings. Decoy allegations won't slow down the system since \sate is scalable.

\section{Security Analysis}
\label{s:security-definition}
To model security and privacy we use the UC framework,
which allows \sate to compose with other cryptographic schemes while maintaining security.

\paragraph{Attacker Model}
Agents (\victims and \tescrows) in our system are interactive Turing machines that communicate with an ideal functionality $\Fsate$.
The adversary $\A$ is a PPT machine with access to an interface $\mathsf{corrupt}(\cdot)$. It takes an agent identifier and returns the internal state of the agent to the adversary. All subsequent incoming and outgoing communication of the agent is then routed through $\A$. The adversary is $f-bounded$, and can corrupt a minority $f < n/2$ of \tescrows and any number of \victims. For formal security, we consider the static corruption model; i.e., the adversary commits to the identifiers of the agents it wishes to corrupt ahead of time.\footnote{The static adversary is a standard assumption employed by most practically relevant MPC systems today.~\cite{SPDZ}}

\paragraph{Communication Model}
We assume the network to be bounded-synchronous~\cite{pederson-mpc} such that the protocol execution occurs in discrete rounds. The agents are aware of the current round, and if a message is created at round $i$, it is delivered at the beginning of round $(i+1)$. Our model assumes that computation is instantaneous. In practice, this is justified by setting a maximum publicly  known time bound on message transmission. If no message is delivered by beginning of the next round, then the message is set to be $\bot$. For an example of the corresponding ideal functionality $\FF_\mathsf{syn}$ we refer the reader to~\cite{CanettiUC, KatzMTZ13}.
The attacker is informed whenever some communication happens between two agents and the attacker can arbitrarily delay the delivery of the message between honest parties within the round boundaries.

The real-world protocol assumes the existence of a functionality $\FF_{anon}$ (see~\cite{formal-onion} for an example), which provides user with an anonymous communication channel.
Moreover, the protocol also assumes the existence of a broadcast channel for allegers to reliably
communicate with all \tescrows and we model this as a bulletin board visible to all \tescrows
(such as~\cite{bulletinboard}) with an ideal functionality
$\FF_\mathsf{B}$.

Concretely, our real-world protocol uses
$\FF_\mathsf{anon}$, $\FF_\mathsf{B}$ and $\FF_\textsf{syn}$ as subroutines,
and is specified in the $(\FF_\mathsf{B},
\FF_\mathsf{anon},\FF_\textsf{syn})$-hybrid model.

\paragraph{Universal Composability}
Let $\mathsf{EXEC}_{\rho, \adv, \mathcal{E}}$ be the ensemble of the outputs of the environment
  $\mathcal{E}$ when interacting with the $f$-bounded adversary $\adv$ and parties
  running the protocol $\rho$ (over the random coins of all the
  involved machines).

\begin{definition}[UC-Security]
A protocol $\rho$ UC-realizes an ideal functionality $\FF$ if
for any adversary $\adv$ there exists a simulator $\mathcal{S}$ such that for
any environment $\mathcal{E}$ the ensembles $\mathsf{EXEC}_{\rho, \adv,
  \mathcal{E}}$ and $\mathsf{EXEC}_{\FF, \mathcal{S}, \mathcal{E}}$ are computationally
indistinguishable.
\label{def:UCsec}
\end{definition}
\begin{figure}[t]
  \begin{mdframed}
    {\bf Initialization}
    \begin{enumerate}[1.]
        \item {\tt registered} $\leftarrow \{\}$, {\tt allegations} $\leftarrow \{\}$, \\
        {\tt buckets} $\leftarrow \{\}$, {\tt unique} $\leftarrow \{\}$
    \end{enumerate}

    {\bf Registration}
    Invoked by client with identity \ID
    \begin{enumerate}[1.]
        \item Send (``Register'', \ID) to all escrows
        \label{fig:ideal-func:reg-reveal}
        \item If received $\bot$ from escrow $i$, then \textsf{IdentifiableAbort}$(i)$
        \item {\tt registered} $\leftarrow$ {\tt registered} $\cup$ \{\ID\}
    \end{enumerate}

    {\bf Allegation Filing}
    Invoked by client with identity \ID, allegation $a$, reveal-threshold $t$, and metadata $m$
    \begin{enumerate}[1.]
        \item If \ID $\notin$ {\tt registered}, then \textsf{Abort}
        \item {\tt registered} $\leftarrow$ {\tt registered} \textbackslash \ID
        \item Send (``File'', \textproc{Unique}(\ID), $t$) to all escrows.
        \label{fig:ideal-func:thresh}
        \item If received $\bot$ from escrow $i$, then \textsf{IdentifiableAbort}$(i)$
        \item {\tt allegations$[\ID]$} $\leftarrow (t, m, a)$
        \item {\tt buckets[$t-1$]} $\leftarrow$ {\tt buckets[$t-1$]} $\cup$ \{($\ID, m$)\}
        \item Send ($t-1$, \textproc{Unique}$(\ID)$, \textproc{Unique}$((t-1, m))$) to all escrows

        \hspace{13mm}{\it \# Matching and Bucketing}
        \item $T \leftarrow$ \textproc{Bucketing}({\tt buckets})
        \item {\bf While} $T \neq \bot$
        \item \quad {\it \# Move allegation $T.\ID$ to bucket $T.i$}
        \item \quad If received $\bot$ from escrow $i$, then \textsf{IdentifiableAbort}$(i)$
        \item \quad $m' \leftarrow$ {\tt allegations}[$T.\ID$].$m$
        \item \quad {\tt buckets$[T.i]$} $\leftarrow$ {\tt buckets$[T.i]$} $\cup$ \{($T.\ID$, $m'$)\}
        \item \quad Send ($T.i$, \textproc{Unique}$(T.\ID)$, \textproc{Unique}(($T.i$, $m'$)))

        \hspace{5mm} to all escrows

        \label{fig:ideal-func:thresh-reveal}
        \item \quad {\tt Task} $\leftarrow$ \textproc{Bucketing}({\tt buckets})

        \hspace{12mm}{\it \# Reveal Allegations}

        \item{\bf For each} $(\ID, M)$ {\bf in} {\tt buckets}$[0]$
        \item\quad {\tt buckets}$[0] \leftarrow\,$ {\tt buckets}$[0]$ \textbackslash $(\ID, M)$
        \item \quad If received $\bot$ from escrow $i$, then \textsf{IdentifiableAbort}$(i)$
        \item \quad $(t, m', a) \leftarrow$ {\tt allegations}$[\ID]$
        \item \quad Send $(t, a, \ID)$ to all \tescrows
        \label{fig:ideal-func:final-reveal}
    \end{enumerate}

    $~$\\

    {\bf function} \textproc{Unique}($x$) {\bf begin}
    \begin{enumerate}[1.]
    \item {\it \# Map input objects to unique numbers}
    \item {\bf if} $x$ {\bf is in} {\tt unique}
    \item \quad {\bf return} {\tt unique}$[x]$
    \item {\bf end if}
    \item {\tt unique}$[x] \leftarrow$ $|${\tt unique}$|$
    \item {\bf return} {\tt unique}$[x]$
    \end{enumerate}

    {\bf end function}

  \end{mdframed}
  \caption{The ideal functionality for \sate, \Fsate. \textproc{Bucketing} is a {\em local} algorithm~(\S\ref{proto:matching:thresholding}) that determines which matches are safe to reveal to the \tescrows.}
  \label{fig:ideal-func}
\end{figure}

\paragraph{Ideal Functionality}
Figure~\ref{fig:ideal-func} describes an ideal functionality $\Fsate$,
which models the intended behavior a \sate, in terms of functionality
and security properties.

For a more modular treatment, our UC definition models only focus on the cryptographic aspects and we assume that all \victims have certificates of real identity from a trusted offline authority. Further, we omit the handling of session IDs (SIDs) in $\Fsate$ to reduce clutter. Messages are assumed to be implicitly associated with SIDs.

In the ideal functionality, {\tt registered} is the set of registered users' identities---only registered users may file an allegation. The set of allegations filed so far is denoted by {\tt allegations}. Each allegation is a tuple of the reveal threshold, meta-data, allegation text and real identity (which is known to the ideal functionality). {\tt unique} counts the number of allegations filed so far, which allows us to assign a unique identifier to each allegation. If an \tescrow (say, the $i^{th}$) refuses to cooperate, \Fsate aborts, reporting the $i^{th}$ \tescrow to the other \tescrows. This is denoted as \textsf{IdentifiableAbort}$(i)$. We assume the adversary is malicious-but-cautious, and wouldn't want to get reported. If authentication fails, then \Fsate calls \textsf{Abort} without a parameter, since it wasn't any \tescrow's fault. In the real protocol, clients can file $l$ allegations for each registration. For notational simplicity, we only show $l=1$ here; but the $l>1$ case is straightforward. Textual description in Section~\S\ref{s:construction} describes the real protocol with $l \ge 1$.

\paragraph{Bucketing Algorithm} To scalably and efficiently implement reveal thresholds, we propose a bucketing protocol~(see~\S\ref{proto:matching:thresholding}) that divides allegations into buckets. All \tescrows know which allegations {\it within} a bucket match each other. This makes the ideal functionality admit a somewhat surprising attack: If an \emph{adversary} files an allegation, it learns whether other matching allegations exist in the same bucket in which the adversary's allegation is placed. These attacks are consistent with our threat model, which allows for probing attacks by adversaries. As explained in \S\ref{proto:matching:thresholding}, the buckets in which an allegation is placed are carefully chosen to disincentivize these attacks by relying on our accountability property~(see~\ref{s:design-requirements}). Note, \textproc{Bucketing} is a local and non-cryptographic algorithm. It merely determines what information can be revealed to the adversary, and hence can be called from the ideal functionality.

\paragraph{Discussion}
$\Fsate$ satisfies the allegation secrecy, \victim anonymity and accountability properties described in~~\S\ref{design-space}, relative to our threat model. {\em Accountability} is ensured since, if a user files an allegation, \Fsate reveals their real identity (\ID) as soon as their threshold is met. Note that we already proved the bucketing protocol correct. {\em Allegation secrecy and \victim anonymity} are ensured because \Fsate reveals information about an allegation only in the following scenarios: (1) \Fsate reveals a user's identity then they register into the system (step~\ref{fig:ideal-func:reg-reveal}).  This is harmless since users register irrespective of whether or not they currently intend to file an allegation. (2) As the bucketing protocol progresses, \Fsate reveals which allegations match which others (step~\ref{fig:ideal-func:thresh-reveal}): we discussed why this information doesn't violate our properties above. (3) It reveals the threshold of an allegation when it is filed (step~\ref{fig:ideal-func:thresh}), hiding which isn't part of our threat model. (4) Finally it reveals the entire allegation when its threshold is met and is ready to be revealed (step~\ref{fig:ideal-func:final-reveal}).
in~~\S\ref{proto:matching:thresholding}, Theorem~\ref{thm:thresh-correct}.

Figure~\ref{fig:sate-real} presents the pseudocode for our cryptographic protocol. We prove UC-security in the $(\FF_\mathsf{B}, \FF_\mathsf{anon}, \FF_\mathsf{syn})$-hybrid model.  Theorem~\ref{thm:uc-security} holds
for any UC-secure realization (as defined in Definition~\ref{def:UCsec}) of $\FF_\mathsf{B}$, $\FF_\mathsf{anon}$ and $\FF_\mathsf{syn}$. We provide a proof sketch of Theorem~\ref{thm:uc-security} in Appendix \ref{sec:uc-security-proof}.

\begin{theorem}[UC-Security]
  \label{thm:uc-security}
Let $\textproc{VSS}$ be a secure verifiable secret sharing scheme, $\textproc{RandomCoinToss}()$ be a secure DKG protocol, let (\textproc{DVRF}, \textproc{VerifyVRF}) be a secure distributed input DVRF protocol, let $H_1$ and $H_2$ be collision resistant hash functions, let $(E,D)$ be a non-committing symmetric encryption scheme, and the employed signature scheme is strongly existentially unforgeable. Then the \sate protocol UC-realizes the ideal functionality $\Fsate$ defined in Figure~\ref{fig:ideal-func} in
the $(\FF_\mathsf{B},
\FF_\mathsf{anon},\FF_\textsf{syn})$-hybrid model.
\end{theorem}

\section{Implementation and Evaluation}
\label{eval}

\paragraph{Implementation}
We build our prototype in Java, with our own implementation of the GRR MPC protocol~\cite{pederson-mpc}. We use SCAPI~\cite{scapi} version~2.3 for establishing communication channels and its bindings to OpenSSL~\cite{openssl} version~1.1, which use use for hashing, symmetric encryption and public-key encryption. We use the Java bindings to the Pairing Based Cryptography library, jPBC~\cite{jpbc-lib} version~2.0 for pairing based cryptography primitives. For operations in $\mathbb{Z}_q$, we use Java BigInteger. To maintain each \tescrow's persistent state, we use a MySQL database to achieve security, scalability and, security. To demonstrate scalability, we pre-populate the database with one million  allegations from one million distinct users. Since our computational complexity per allegation/registration does not depend on the number of pre-existing allegations/registrations, this imposes a negligible overhead on the protocol.

We now evaluate our implementation to show that \sates are fast enough for practical use.

\paragraph{Latency and throughput}
We first measure the latency and throughput of user-\sate interaction
in a realistic setting. We set up to 9 \tescrows on Amazon AWS cloud servers,
chosen to maximize geographical extent. In an experiment involving $n$
\tescrows, the \tescrows run on servers in the first $n$ of Virginia,
Frankfurt, Sydney, N. California, Singapore, Sao Paulo, London, Seoul,
and Mumbai. Each \tescrow runs on a M4.large AWS instance. At the time
of the experiments, this provided 2 vCPUs, 8GB of RAM, and `moderate'
network performance. Each server runs up to 60 threads, the maximum
supported on the machines; each thread handles one concurrent client
request. Note that the \sate registration is embarrassingly parallel with
respect to client requests---cost is dominated by network latencies
and MPC computation, which require no syncing across client requests;
such synchronization is needed only storing registered identities to database.
Allegation filing must be done serially one-by-one. We use up
to 60 client replicas, all hosted on a single c4.4xlarge instance of
AWS in Virginia. At the time of our experiments, this provides 16 vCPUs,
30GB RAM and `High' network performance.

\emph{Latency:} Figure~\ref{fig:real-eval} (top) shows the average
latency for registering a new key as the number of \tescrows varies,
in two configurations: When the \tescrows are lightly loaded (no
concurrent requests) and when they are heavily loaded (60 concurrent
clients). Latency is the time between when a user sends its request, to when it gets the \sate's MACs on its keys. There are three notable aspects here. First, as expected,
the latency increases with the number of \tescrows (since the MPC
becomes more complex). Second, increasing the number of concurrent
clients does not increase the latency significantly. This suggests
that the cost is \emph{dominated} by the number of \tescrows and
inter-escrow network latencies. Finally, even though the absolute
latency numbers might look high (of the order of 10s of seconds), they
are acceptable since user interaction with \sates is relatively
infrequent. In particular, users register new keys once every few
months, so such latencies seem quite practical. Latency is not a concern for filing an allegation, since the user does not expect any response from the \tescrows. The cost of matching and bucketing is better captured in terms of throughput.

\emph{Throughput:} Next, we measure the throughput of \sate in terms
of the number of key registrations and allegation filings it can
handle per second. For registration, we use 60 concurrent clients. Figure~\ref{fig:real-eval} (bottom) shows the throughput as a function of the number of \tescrows. Allegations are filed serially. As expected, the throughput number decreases with increasing number of \tescrows.

For allegation filing, each client repeatedly files allegations
with thresholds varying between 2 and 20, chosen from a truncated
exponential distribution with mean 5. When a threshold of $t$ is
chosen, $t$ matching allegations are created with 50\% probability,
and $t-1$ matching allegations are created the rest of the
time. These, respectively, represent the cases where the allegation is
eventually revealed and the worst-case (for performance) when the
allegation is not actually revealed.

We believe that these throughputs are acceptable for \sate, since user
operations are expected to be very infrequent. Moreover, each \tescrow
can be separately replicated on several servers to get proportionally
higher throughput.

\begin{figure}[!t]
  \includegraphics[width=\linewidth]{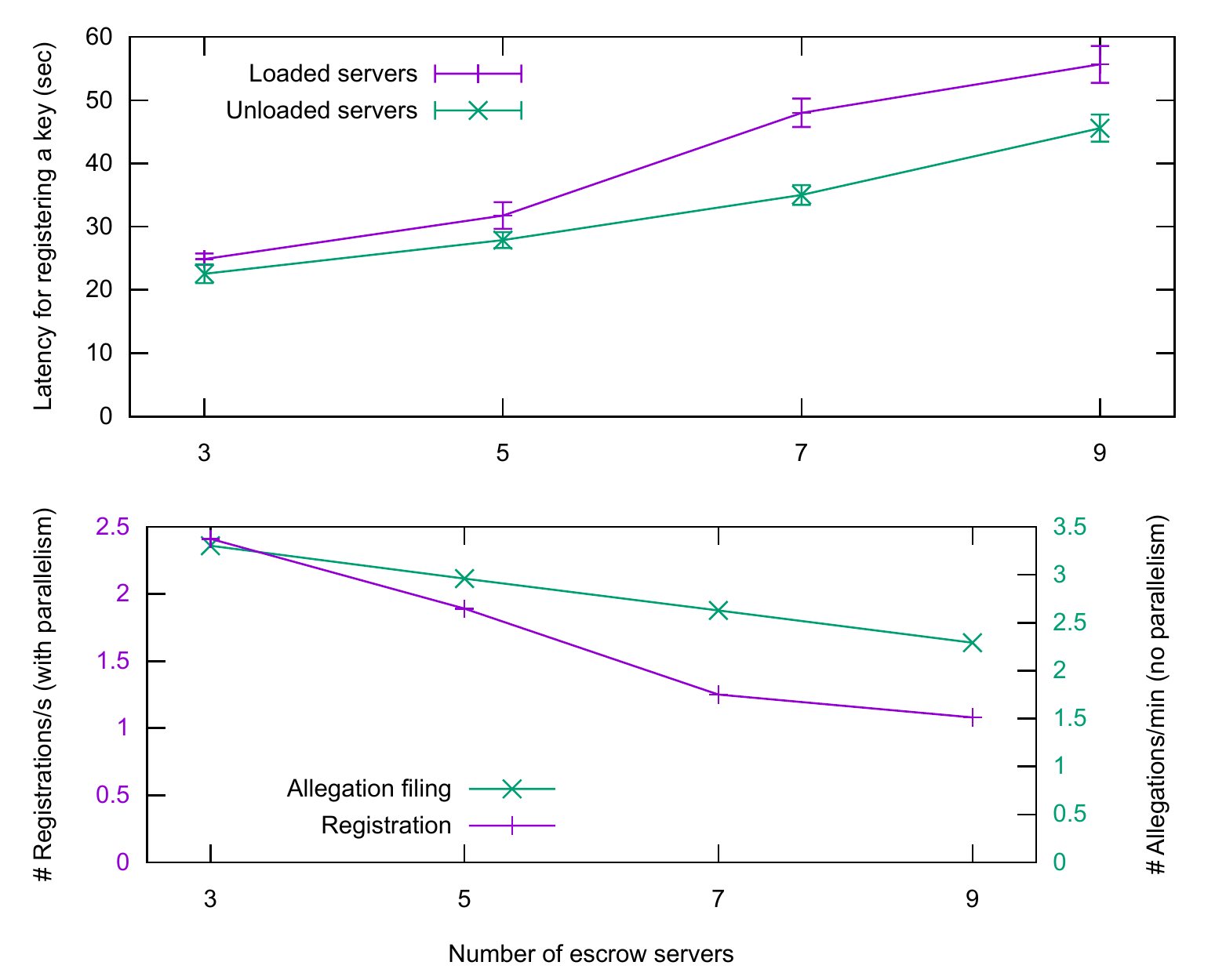}
  \caption{Latency and throughput when the \sate protocol is deployed on \tEscrow
    servers running on AWS instances in different locations. Registrations occur in parallel, while allegations are filed serially. Hence throughput is different for the two (note the units: seconds vs minutes).}
  \label{fig:real-eval}
\end{figure}

\begin{figure}[!t]
  \includegraphics[width=\linewidth]{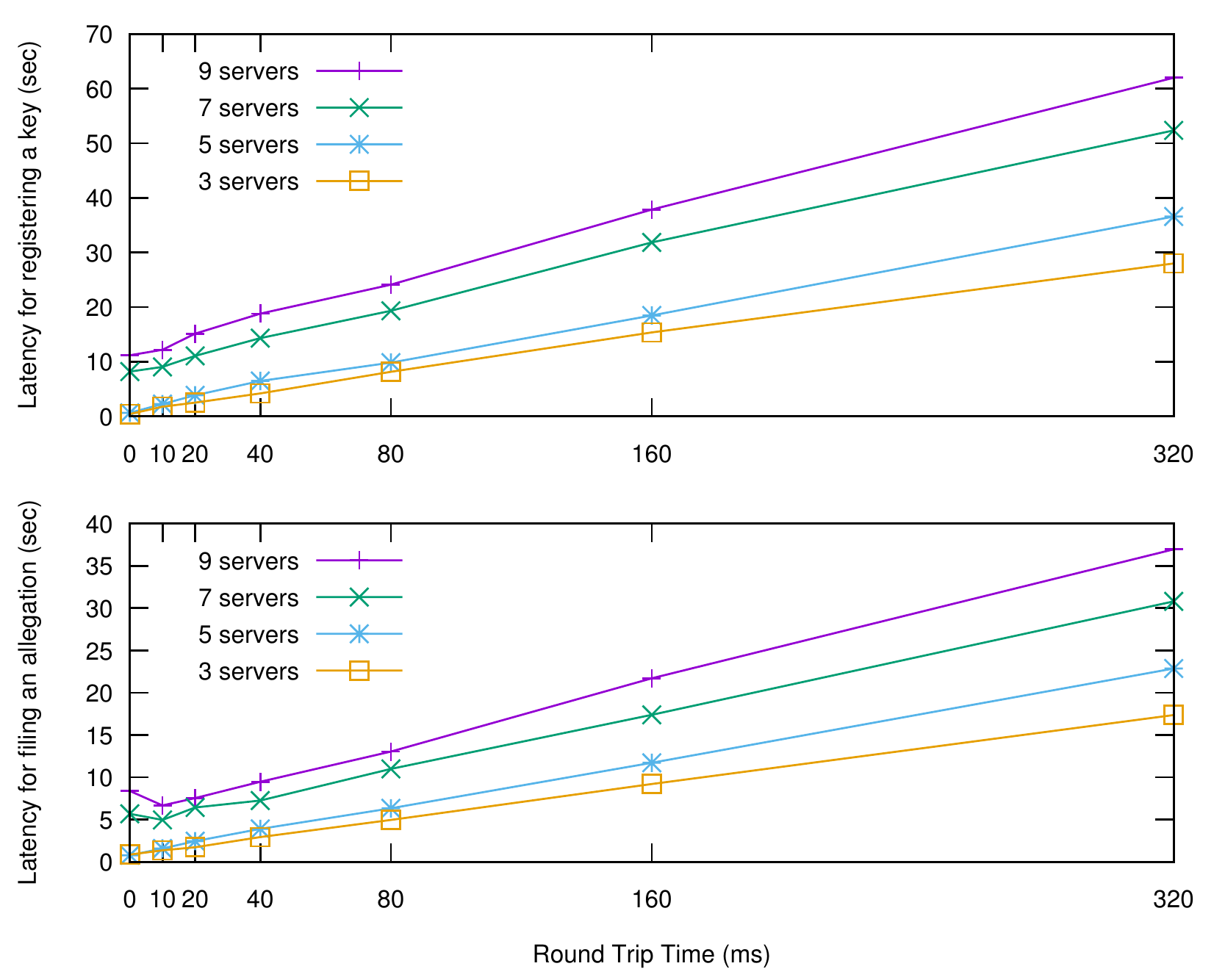}
  \caption{End-to-end latency for registering a key and filing an allegation as
    the communication delay between every pair of \tEscrows on an emulated
    network is varied. Values are shown for number of \tEscrows varying from 3
    to 9.}
  \label{fig:vary-rtt}
\end{figure}

\paragraph{Impact of network latency}
The primary source of user-perceived latency is inter-escrow network latency. To test this, we emulate a network with Linux qdiscs to get predictable performance on a single Amazon
AWS c4.4xlarge instance. The \tescrow servers and our client occupy
one core each. Every pair of \tescrow servers is given an emulated 100
Mbps link and 1 bandwidth$\times$delay worth of buffer (the recommended buffer size for TCP to obtain full link
utilization and minimal delay). We vary the latency of the
emulated network links and Figure~\ref{fig:vary-rtt} plots the latencies of a) registering a key,
and b) \emph{processing} one allegation completely  with no
matches. These require two and one VRF computations,
respectively. (Note that the user perceived latency of allegation
\emph{filing} is different from that of processing the
allegation. Filing does not require any PRF computation.) As expected, the client
latency increases linearly with the network latencies, and the rate of
increase also increases with the number of \tescrows.

\section{Discussion}
\subsection{Deployment Considerations}

\paragraph{Client Software} Users need client software to participate in the
protocol. It would be convenient if this software were part of a web-page.
However, it is challenging to access anonymity services like ToR from within a
browser (which suggests interesting future work). If users were required to
download special software instead, we expose other security concerns. The very
act of downloading the software could indicate an intent to file an allegation,
and not all users will use an anonymity service like ToR to do so. To prevent
this channel of inference, the client software should be bundled with other
commonly used software. Alternately, we can produce cover traffic by making a
fraction of all visitors of a popular web-page (e.g. organization's home page)
download the software.

\paragraph{Practical Security} Like all software, \sate will have security
vulnerabilities; there is no use in encrypting secrets codes a buffer-overflow
attack leaks the secret keys. In addition to careful code audits, we could make
multiple implementations that use independent hardware/software stacks and
compare their outputs to see they are identical. If not, we halt the system
until a security expert can find and fix the bug. This forces an attacker to
find the {\it same} vulnerability on different hardware/software stacks, which
is much more difficult. Such a heavy-handed approach is prudent here since
security is much more important than performance and availability.

\paragraph{Non-technical considerations} Since \sate handles sensitive
information, its social design requires considerable thought. While a full
analysis is out-of-scope for this primarily technical paper, we discuss some
issues here. When thresholds for a set of allegations are met, to whom should
they be revealed? To avoid centralization, we could reveal them to the allegers
themselves. They can then coordinate and report to school authorities or law
enforcement if they so choose. On the other hand, if allegations went to a central
authority first, it could provide legal assistance to each alleger {\it
before} introducing them to each other. It could also filter out any obivously
fake/probe allegations. To keep this authority accountable, the escrows could
notify allegers that their allegation has been matched, but not to whom. Project
Callisto~\cite{project-callisto} currently follows a similar approach.

\sate forces allegers to create a paper-trail while filing an allegation to
discourage fake/probe allegations. We need effective legal mechanisms to make
this a significant deterrent. for more discussion of other social issues, we
refer the reader to prior social-science work~\cite{information-escrow-law,
project-callisto-report17}.

\subsection{Future Work} \paragraph{Withdrawing/modifying allegations} A user
can readily update allegation free-form text by sharing a new value. However,
\sate cannot always let them withdraw an allegation or modify its
metadata/reveal threshold. Since an allegation's threshold could be met as soon
as it is filed, allegers should only file one if they are comfortable with it
being revealed. Nevertheless, allegers may want to modify one. For example if an
allegation hasn't been revealed in several months/years, they may want to either
withdraw it or reduce its reveal threshold. In \sate, this is hard, since for
thresholds $> 2$, escrows may know that two allegations match even before they
are revealed. An adversary can use this to probe for allegations; if they can
withdraw their probe allegation, they can delete the paper trail that
disincentivizes such probes~(see~~\S\ref{proto:matching:thresholding}).

Note, if an allegation isn't yet known match any another, we can allow it to be
withdrawn. To do so, we must pick a new secret key for the bucket it is in and
recompute all PRFs with this new key. Since this is computationally expensive,
we can limit the overhead by recomputing PRFs at-most once per week. Users are
notified that withdrawal can take upto one week, which is acceptable in this
context. Fully supporting allegation withdrawal and modification is interesting
future work.

\paragraph{Other matching criteria} \sate matches allegations based on exact
string equality. Could we support other criteria? For instance, some allegers
may be victims of the crime and others may be witnesses. Can we support
different thresholds based on type of alleger? Could we match on multiple
fields, e.g. `match only if at-least two of (name/phone number/email
address/employee id) match', `match against this accused person only if the
crime happened within this time-frame/physical coordinates', or `match only
against allegations filed in the last year'. In applications where some
ambiguity is acceptable, it would be interesting to match based on softer
criteria provided by machine learning. E.g. based on a person's picture,
or a textual description. Note, these more complex criteria break transitivity.
That is `$A$ matches $B$' and `$B$ matches $C$' doesn't imply `$A$ matches $C$'.
Future work would need to define what it means for the thresholds for a group of
allegations to be satisfied.

\paragraph{Identity Management} To enable real-world identities, allegation
escrows require a robust public-key infrastucture (PKI), for which users should
validate their identity with a trusted authority. If a user registers
immediately before filing an allegation, then this act reveals an intent to file
an allegation. Hence the PKI must be established beforehand. \sate requires a
pre-registration step in addition to PKI, which is acceptable in many cases
since a new PKI must be established anyway; most organizations either don't have
one for their employees/students, or the ones they have aren't very robust. For
instance, some administrators may have access to employee logins/emails.
Nevertheless, it is interesting future work to explore how to effectively
exploit a pre-existing PKI to avoid a separate pre-registration for the escrow
system. For instance, if each person had a certificate, they could secret share
their identity and prove in zero-knowledge that they own the secret-shared
identity. Prior work demonstrates how to do this while being backwards
compatibile with X.509 certificates~\cite{cinderella} and email
services~\cite{email-blind-certificate}. Future work must establish that this is
practical and efficient. Additionally, many organizations use multi-factor
authentication. Exploiting this additional layer of security is also interesting
future work.

\section{Conclusion}
\label{conclusion}

We have presented \sate, a robust system that implements an allegation
escrow with strong cryptographic security guarantees, and showed that
it is practical. \sate keeps allegations and the identities of
\victims and the \accused confidential until \victim-specified
match-thresholds are reached. The system's security and privacy
guarantees provably hold as long as a majority of the escrow parties
are uncorrupted. Our empirical evaluation suggests that \sates are
efficient enough to be used in practice, and scales well to large numbers of users and allegations.

\section{Acknowledgements}

We thank the anonymous reviewers. We are also grateful to Krishna Gummadi, Prabhanjan
Ananth, Hari Balakrishnan, Derek Leung, Omer Paneth, Malte Schwarzkopf, Frank
Wang and Nickolai Zeldovich for many interesting discussions and valuable
feedback on the paper.

\bibliographystyle{IEEEtranS}
\bibliography{escrow,crypto}
\appendix
\subsection{Bilinear Pairings}
\label{app:crypto:bilinear}
Let $\mathbb{G}_1, \mathbb{G}_2, \mathbb{G}_T$ be multiplicative, cyclic groups
of prime order $q$. Let $g_1, g_2$ be generators of $\mathbb{G}_1, \mathbb{G}_2$
respectively. A map $e : \mathbb{G}_1 \times \mathbb{G}_2 \rightarrow
\mathbb{G}_T$ is called bilinear if it has the following properties. (1) {\bf
  Non-degenerate: } $e(g_1, g_2) \neq 1$. (2) {\bf Bilinear: } For all $u \in
\mathbb{G}_1, v \in \mathbb{G}_2, x, y \in \mathbb{Z}$, $e(u^x, v^y) = e(u,
v)^{xy}$. (3) {\bf Computable: } There is an efficient algorithm to compute
$e(u, v)$ for all $u \in \mathbb{G}_1, v \in \mathbb{G}_2$. For ease of
exposition, we assume that the pairing employed is symmetric, i.e.,
$\mathbb{G}_1 = \mathbb{G}_2 = \mathbb{G}$~\cite{GPS'08,DBDH}. Unless mentioned otherwise, $g \in \mathbb{G}$ is a publicly known generator.

\subsection{Verifiable Pseudorandom Functions (VRFs)}
\label{app:crypto:vrf}
VRFs cannot be distinguished from a random function by a computationally
bounded adversary that does not have access to the proof. For our purposes, we
adopt the following formal definition of a VRF from~\cite{pairing-vrf}.
Let $a_1 : \mathbb{N} \rightarrow \mathbb{N} \cup \{*\}$ and $a_2 :
\mathbb{N} \rightarrow \mathbb{N}$ be functions computable in $\mathrm{poly}(k)$
time\footnote{Except when $a_1$ takes the value $*$, which means the VRF is
  defined for inputs of all length.}. $F_{(\cdot)}(\cdot) : \{0,1\}^{a(\lambda)}
\rightarrow \{0,1\}^{b(\lambda)}$ is a family of VRFs if there exists a PPT
(probabilistic polynomial time computable) algorithm $GEN$ and deterministic
algorithms $PROVE$ and $VER$ such that $GEN(1^\lambda)$ outputs a pair of keys
$(SK, PK)$; $PROVE_{SK}(x)$ computes $(F_{SK}(x), \pi_{SK}(x))$, where
$\pi_{SK}(x)$ is a proof of correctness; and $VER_{PK}(x, y, \pi)$ verifies that
$y = F_{SK}(x)$. They satisfy the following properties: 1) {\it Uniqueness:} No
values $(PK, x, y_1, y_2, \pi_1, \pi_2)$ satisfy $VER_{PK}(x, y_1, \pi_1) = 1 =
VER_{PK}(x, y_2, \pi_2)$ when $y_1 \neq y_2$, 2) {\it Provability:} If $(y, \pi)
= PROVE_{SK}(x)$, then $VER_{PK}(x, y, \pi) = 1$ and, 3) {\it Pseudorandomness:}
For any PPT algorithm $A = (A_1, A_2)$ that does not query its oracle on $x$,
the following holds
\begin{displaymath}
  Pr\left[b = b' \left|
    \begin{array}{l}
      (SK, PK) \leftarrow GEN(1^\lambda),\\
      (x, st) \leftarrow A_1^{PROVE(\cdot)}(PK), \\
      y_0 \leftarrow F_{SK}(x), y_1 \leftarrow \{0,1\}^{a_2(\lambda)}, \\
       b = \{0,1\}, b' = A_2^{PROVE(\cdot)}(y_b, st)
    \end{array}
    \right. \right] \le \frac{1}{2} + \negl(\lambda)
\end{displaymath}
where $\negl(\cdot)$ is a negligible function. Further, it satisfies the following unpredictability property

For any PPT algorithm $A$, who does not query its oracle on $x$, the following
holds:
\begin{displaymath}
  Pr\left[y = F_{SK}(x) |
    \begin{array}{l}
      (PK, SK) \leftarrow GEN(1^\lambda);\\
      (x,y) \leftarrow A^{PROVE(\cdot)}(PK)
    \end{array}
      \right] \le negl(k)
\end{displaymath}

\subsection{Postponed Security Analysis}
\label{sec:uc-security-proof}

\begin{definition}
	\label{def1} ~\cite{Chou2015}.
A symmetric encryption scheme $(E, D)$ is non committing if there exist two PPT algorithms $(\mathcal{A}_1, \mathcal{A}_2)$ s.t.   $(c, k)$ and $(c', k')$ are computationally indistinguishable when $c' \leftarrow \mathcal{A}_1(1^\lambda)$, $k' \leftarrow \mathcal{A}_2(c', M)$, $k \leftarrow \mathcal{K}$ and $c \leftarrow E(k, M)$ for all $M \in \mathcal{M}$ where $\mathcal{K, M, C}$ denote key, message and ciphertext spaces respectively
\end{definition}

We refer \cite{Chou2015} for a simple construction.

\begin{proof}[Proof Sketch for Proof Theorem~\ref{thm:uc-security}]
Our proof strategy consists of the description of a simulator $\sim$ that handles users corrupted by the attacker and simulates the real world execution protocol while interacting with the ideal functionality \Fsate.

The simulator $\sim$ spawns honest users at adversarial will and impersonates them until the environment $\environment$ makes a corruption query on one of the users: At this point $\sim$ hands over to $\adv$ the internal state of the target user and routes all of the subsequent communications to $\adv$, who can reply arbitrarily. For operations exclusively among corrupted users, the environment does not expect any interaction with the simulator. Similarly, interactions exclusively among honest nodes happen through secure channels and therefore the attacker does not gather any additional information other than the fact that the interactions took place. For simplicity, we omit these operations in the description of our simulator.  The simulator simulates the following honest nodes: 1) the honest \tescrows, 2) the honest users, 3) the CA for users' real identities.
Next, we describe how the simulator behaves at various points of the protocol.

At several points in the \sate protocol, DKG is required. namely, $SK_I$ used to compute MACs on identities, $SK_R$ used for revealing alleger identity and $SK_i$ for each $i^{th}$ bucket used for thresholding. To simulate this with a minority of statically corrupted \tescrows, $\sim$ chooses a random key pair, performs DKG simulation~\cite[Theorem~1]{DKGJournal}, and sends the the public key to the corrupted \tescrows. As this simulation is exactly the distribution in the real protocol~\cite[Theorem~1]{DKGJournal}, and hence is indistinguishable from it. Notice that the simulator knows all the DKG secret keys here. It participates in computing $PK_I$ from $SK_I$. The simulator also generates the public-private key pairs for all the honest users and generates certificates for them from the CA.

For allegation filing and registration, we consider two cases depending on whether or not the alleger is honest.

\smallskip\noindent{\bf Case~1: Honest alleger, corrupted minority of \tescrows}

When an honest alleger registers, $\Fsate$ sends $(\register, \ID)$ to the the simulator. The simulator proves the honest alleger's identity to the corrupted \tescrows. This is possible because it simulates the CA and can generate arbitrary certificates. Then it generates $l$ new public keys $pk_1, \ldots, pk_l$ (note, figure~\ref{fig:sate-real} shows only $l=1$ for notational simplicity) and secret shares them among the \tescrows and participates in the distributed computation of $\pi_{SK_I}(H_2(pk_i))$ and $F_{SK_R}(H_2(pk_i))$ as described in~\S\ref{sec:Distributed} (note, the simulator knows $SK_I$ and $SK_R$). If the adversary refuses to participate in this computation, the simulator sends \fail to \Fsate from a corrupted escrow's channel and aborts. Else it sends {\tt OK}.\footnote{Henceforth, whenever the adversary makes the $i^{th}$ escrow fail, the simulator sends \fail from that escrow's channel. But we omit this detail for clarity.}. As in the real protocol, the adversary obtains $F_{SK_R}(pk_i)$, but not $\pi_{SK_I}(pk_i)$. So far, this is exactly what happens in the real protocol, except that DKG and the honest parties' private keys are chosen by the simulator, but from the same distribution. Hence it is indistinguishable from the real execution.

When an honest alleger files an allegation, $\Fsate$ sends $(\mathrm{``File"}, UID, t)$ to the simulator. The simulator chooses a random key-pair $(sk,pk) \leftarrow \mathrm{Gen}(1^\lambda)$, generates a MAC $\pi_{SK_I}(H_2(pk))$ on it and sends $("\mathrm{File}", pk, t)$ and $(\pi_{SK_I}(H_2(pk)), \mathcal{C})$ to the corrupted escrows signed using $sk$, where the $\mathcal{C}$ is a random non-committing encryption ciphertext. The simulator generates a random meta-data $m = H_1(meta-data)$ and symmetric key $k$, and distributes a minority of shares among the corrupted escrows as $VSS(m)$ and $VSS(k)$, signed with $sk$. The distribution of meta-data doesn't matter since it is information theoretically hidden from the adversary. Since the adversary has not seen the honest alleger's public key before, the simulator can choose a random one.
$\Fsate$ now moves to matching and thresholding, returning $(i, UID, Um)$ each time an allegation identified by $UID$ is added to bucket $i$. Let $pk$ be the public key the simulator chose for $UID$ (in the dishonest \victim case discussed below, the adversary provides $pk$, corresponding to which $\Fsate$ provides $UID$).

At this point, the real protocol would be computing $F_{SK_i}$({\tt allegations}$[pk].m$). The simulator can control the value of this result. If $Um$ matches any other allegations in bucket $i$, the simulator produces the value it previously returned for that allegation in bucket $i$. Else it produces a fresh random value. This works because $H_1$ is collision resistant, $H_1(m)=H_1(m')$ iff $m=m'$ for a computationally bounded adversary. Since $F$ is a PRF, the adversary cannot distinguish between its output and truly random numbers. Note all matching allegations have the same (hash of) meta-data $m$ by definition.  If at any point, the adversary refuses to cooperate in distributed-input DPRF computation, the protocol is aborted, and the simulator sends \fail to $\Fsate$, which also halts execution. Else it sends OK each time to move the protocol forward.

To reveal identity in the real protocol, the \tescrows compute $F_{SK_R}(H_2(pk))$, where $pk$ was the public key used during allegation. To simulate this, the simulator picks $pk_i$ randomly from the set of unrevealed public keys it chose when {\tt ID} was registered. It simulates the other \tescrows' behavior such that, if the adversary cooperates, it gets $F_{SK_R}(H_2(pk_i))$. Note, the simulator knows $SK_R$.  The simulator sends shares of the (non-committing) symmetric encryption key from honest \tescrows such that the ciphertext $C$ open to $a$ to the corrupted \tescrows. Allegation reveal now succeeds.

\smallskip\noindent{\bf Case 2: Corrupted alleger, corrupted minority of \tescrows}

During registration, the adversary provides a proof of $\ID$ from a CA to the simulator. It also sends the honest \tescrows' shares of hashes of $l$ public keys $H_2(pk_1),\ldots,H_2(pk_l)$ to the simulator. If the proof of $\ID$ is invalid, or the shares are incorrect (i.e. VSS verification fails), the simulator sends \fail to adversary. Else, it sends $(\register, ID)$ to \Fsate from the corrupted allegers' $\ID$. Note, the simulator has a majority of shares of $H_2(pk_i)$ and can hence reconstruct them. It also knows the secret keys $SK_I$ and $SK_R$. Hence it can participate in the computation of $\pi_{SK_I}(H_2(pk_i))$ and $F_{SK_R}(H_2(pk_i))$ on the $l$ public keys to produce the correct result. If the adversary refuses to participate in the computation, it sends \fail to \Fsate.

When filing an allegation, the alleger sends $(t, pk, \pi_{SK_I}(H_2(pk)), Enc_k(a))$ to the simulator for broadcasting. It secret shares the key $k$ and a collision-resistant hash of the meta-data, $m$. The simulator verifies that $pk$ has not been used before and verifies the MAC on it. If the check fails, the simulator sends \fail from the honest escrows to the corrupted alleger. If verification succeeds, the simulator determines the {\tt ID} with which $pk$ was registered (since it has all the registered keys), and connects to $\Fsate$ from {\tt ID}'s channel. It then invokes registration with $\Fsate$ with $(\mathrm{File}, m, a, t)$, which responds with $(\mathrm{File}, C, t)$\footnote{The simulator knows $(m, a)$ since it has a majority of the necessary shares. Again, if the shares are invalid, it sends \fail to the adversary as verifiable secret-sharing is used.}. Now the bucketing algorithm takes place, the simulation process for which is identical to the honest alleger case. $\Fsate$ returns matching allegations for various buckets, and we simulate for the corrupted escrows, a pseudo-random function on the meta-data. This is possible since we know, for the relevant buckets, meta-data of which allegations match.

When an allegation filed by a corrupted party is to be revealed, $\Fsate$ sends (``Reveal'', $C, t$) to the simulator. The simulator cooperates in computing $\pi_{SK_R}(pk)$, where $pk$ is the corresponding key used to file the allegation identified by $C$. If the adversary refuses to cooperate, the simulator sends \fail to \Fsate. Else, it sends {\tt OK}, and cooperates to reveal $a$, which it knows.
\end{proof}

\balance


\end{document}